%% file: main.tex
\documentclass{article}

\usepackage{arxiv}

\usepackage[utf8]{inputenc} 
\usepackage[T1]{fontenc}    
\usepackage{hyperref}       
\usepackage{url}            
\usepackage{booktabs}       
\usepackage{amsfonts}       
\usepackage{nicefrac}       
\usepackage{microtype}      
\usepackage{lipsum}		
\usepackage{graphicx}
\usepackage{natbib}
\usepackage{doi}

\title{Deriving Compact QUBO Models via Multilevel Constraint Transformation}

\author{ Oksana Pichugina,\; Yingcong Tan,\; J. Christopher Beck\\
	Department of Mechanical and Industrial Engineering\\University of Toronto\\Toronto, Ontario M5S 3G8, Canada \\
	\texttt{\{opich, yctan, jcb\}@mie.utoronto.ca} \\
}

\renewcommand{\shorttitle}{Deriving Compact QUBO Models via Multilevel Constraint Transformation}

\hypersetup{
pdftitle={A template for the arxiv style},
pdfsubject={q-bio.NC, q-bio.QM},
pdfauthor={David S.~Hippocampus, Elias D.~Striatum},
pdfkeywords={First keyword, Second keyword, More},
}

\usepackage{graphicx}%
\usepackage{multirow}%
\usepackage{amsmath,amssymb,amsfonts}%
\usepackage{amsthm}%
\usepackage{mathrsfs}%
\usepackage[title]{appendix}%
\usepackage{xcolor}%
\usepackage{textcomp}%
\usepackage{manyfoot}%
\usepackage{booktabs}%
\usepackage{algorithm}%
\usepackage{algorithmicx}%
\usepackage{algpseudocode}%
\usepackage{listings}%
\usepackage{subcaption}
\usepackage{array}
\usepackage{booktabs}

\newcommand{\bbR}{\mathbb{R}}
\newcommand{\bbZ}{\mathbb{Z}}
\newcommand{\bbB}{\mathbb{B}}

\newcommand{\bx}{\mathbf{x}}
\newcommand{\by}{\mathbf{y}}
\newcommand{\bs}{\mathbf{s}}
\newcommand{\ba}{\mathbf{a}}
\newcommand{\bA}{\mathbf{A}}
\newcommand{\bb}{\mathbf{b}}
\newcommand{\bB}{\mathbf{B}}
\newcommand{\bc}{\mathbf{c}}
\newcommand{\bC}{\mathbf{C}}
\newcommand{\bg}{\mathbf{g}}

\newcommand{\bP}{\mathbf{P}}
\newcommand{\bQ}{\mathbf{Q}}
\newcommand{\bW}{\mathbf{W}}

\newcommand{\cH}{\mathcal{H}}

\DeclareMathOperator*{\minimize}{\text{minimize}}

\DeclareMathOperator*{\subjto}{\text{subject to}}

\newcommand{\ie}{\textit{i.e., }}
\newcommand{\eg}{\textit{e.g., }}
\newcommand{\ignore}[1]{}

\newtheorem{theorem}{Theorem}
\newtheorem{lemma}[theorem]{Lemma}
\newtheorem{corollary}[theorem]{Corollary}
\newtheorem{proposition}[theorem]{Proposition}%
\newtheorem{remark}{Remark}%

\newtheorem{definition}{Definition}%

\begin{document}
\maketitle

\begin{abstract}
With the advances in customized hardware for quantum annealing and digital/CMOS Annealing, Quadratic Unconstrained Binary Optimization (QUBO) models have received growing attention in the optimization literature.
%
Motivated by an existing general-purpose approach that derives QUBO models from binary linear programs (BLP), we propose a novel \emph{Multilevel Constraint Transformation Scheme} (MLCTS) that derives QUBO models with fewer ancillary binary variables.
We formulate sufficient conditions for the existence of a compact QUBO formulation (\ie in the original BLP decision space) in terms of constraint levelness and demonstrate the flexibility and applicability of MLCTS on synthetic examples and several well-known combinatorial optimization problems, \ie the Maximum 2-Satisfiability Problem, the Linear Ordering Problem, the Community Detection Problem, and the Maximum Independence Set Problem.
For a proof-of-concept, we compare the performance of two QUBO models for the latter problem on both a general-purpose software-based solver and a hardware-based QUBO solver. The MLCTS-derived model demonstrate significantly better performance for both solvers, in particular, solving up to seven times more instances with the hardware-based approach.
\end{abstract}




\input{sec1-intro}

\input{sec2-background}

\input{sec3-method}

\input{sec4-mlcts-constraint}

\input{sec5-mlcts-cop}

\input{sec6-discussion}

\input{sec7-conclusion}

\bibliographystyle{unsrtnat}
\bibliography{reference}  

\clearpage
\input{sec8-appendix}

\end{document}

%% file: sec1-intro.tex
\section{Introduction}

\emph{Quadratic Unconstrained Binary Optimization} (QUBO) problems consist of a quadratic objective function over a set of binary variables \cite{fu1986}. While QUBOs have numerous applications in artificial intelligence \cite{Cohen20,sleeman2020,willsch2020svm,Zhang22}, operations research \cite{alidaee2008new,alidaee1994,kochenberger2005,Feld19,Senderovich22a}, economics \cite{hammer1971applications}, and finance \cite{cohen2020portfolio,elsokkary2017financialprofolio,laughhunn1970quadratic}, recent development of customized hardware for optimization of QUBO models (e.g., Quantum Annealing and digital/CMOS annealing) have further increased research interest in QUBO modeling and solving \cite{Ajagekar2020,DA3,tran2016,Venturelli19}.

To take advantage of these novel hardware platforms, it is necessary to derive QUBO models for the problem of interest. A standard approach starts with a constrained binary linear model and transforms the constraints into quadratic penalties that are then added to the QUBO objective function \cite{punnen_quadratic_2022}. 
%
Binary optimization problems themselves have a wide range of applications in engineering \cite{domino_quadratic_2022}, computer science \cite{verma_efficient_2021}, operations research \cite{alidaee2008new,alidaee1994,Feld19,kochenberger2005,Senderovich22a}, artificial intelligence \cite{Cohen20,sleeman2020,willsch2020svm,Zhang22}, economics \cite{hammer1971applications}, finance \cite{cohen2020portfolio,elsokkary2017financialprofolio,laughhunn1970quadratic} and many other fields. Of course, generally, binary optimization problems are computationally intractable  \cite{fu1986,punnen_quadratic_2022,garey_computers_1979} and this combination of applicability and intractability explains the ongoing interest in the development of solution approaches.

In this paper, we propose a novel QUBO modelling scheme, the \emph{Multilevel Constraint Transformation Scheme} (\emph{MLCTS}), that transforms a binary linear program formulation into a QUBO formulation. 
We demonstrate the flexibility and applicability of the MLCTS on several well-known combinatorial optimization problems and derive novel QUBO formulations. 
As a proof-of-concept, we numerically evaluate the performance of different QUBO formulations of the Maximum Independent Set Problem (MIS) on both a general-purpose commercial solver \cite{gurobi} and a specialized QUBO solver, the Fujitsu Computing as a Service Digital Annealer \cite{DA3}. Our numerical results show a significant improvement in performance of both solvers, substantially increasing the number of instances solved within a given time limit. For example, seven times more MIS instances can be solved with the Fujitsu Digital Annealer in 5 seconds when using our novel QUBO formulation compared to what it can solve with the standard QUBO model.
Lastly, we  discuss how the MLCTS may be extended to transform a general binary optimization problem model into a QUBO model and illustrate the transformation on two synthetic examples. 



The paper is organized as follows. The background is discussed in Section \ref{sec: Background} and the MLCTS is introduced in Section \ref{sec: Method}. Section \ref{sec: MLCTS developing} demonstrates how MLCTS can be used to derive compact penalties from various constraints. 
MLCTS is applied to several well-known combinatorial optimization problems in Section \ref{sec: applied problems} and we provide a proof-of-concept numerical experiment in Section \ref{sec: MIS experiment results}. Section~\ref{sec: discussion} presents the discussion and Section \ref{sec: conclusion} concludes the paper.

%% file: sec2-background.tex
\section{Background}\label{sec: Background}
 
\subsection{Binary Optimization Problems}



A \emph{binary optimization problem} (BOP) can be represented in the following general form.

\begin{subequations}
\begin{align}
(\text{BOP})\quad f^*=f_0(\bx^*)=\min_{\bx}  &\quad f_0(\bx)\label{eq: BOP objective}\\
\subjto &\quad f_i(\bx) =b_i,\: i\in \overline{1,m}\label{eq: BOP eq constr}\\
&\quad f_i(\bx) \le b_i,\: i\in \overline{m+1,m+m'}, \label{eq: BOP ineq constr}
\end{align}
\end{subequations}
where $\bbB^n=\{0,1\}^n$, $\bx\in\bbB^n$ is a  $n$-dimensional vector of binary decision variables and $f_i: \bbB^n \to \bbR, i\in \overline{0, m+m'}$.\footnote{$\overline{n,m} =\{n,\: n+1,\: \dots,\: m-1,\: m\}$ is the collection of all integer values in the range of $n$ and $m$ (inclusive). 
}
%

The most studied class of BOP is the \emph{binary linear program} (BLP) where constraints \eqref{eq: BOP eq constr} and \eqref{eq: BOP ineq constr} 
are linear.
We consider a BLP formulation in the following form. 
\begin{subequations}\label{eq: LBP Model}
\begin{align}
(\text{BLP})\quad f^*=f(\bx^*)=\min_{\bx}  &\quad \{f(\bx)=\ba_0^\top \bx\}\label{eq: LBP objective}\\
\subjto &\quad \bA\bx =\bb, \label{eq: LBP eq constr}\\
&\quad \bA'\bx \le \bb', \label{eq: LBP ineq constr}
\end{align}
\end{subequations}
 where  $\ba_0\in \bbR^{n}$, $A\in \bbR^{m\times n}$,  $A'\in \bbR^{m'\times n}$,  $\bb\in \bbR^{m}$, $\bb'\in \bbR^{m'}$. 


\subsection{Quadratic Unconstrained Binary Optimization} \label{sec:related_work:qubo}
The \emph{Quadratic Unconstrained Binary Optimization} (QUBO) problem can be stated as 
\begin{equation} \label{eq: qubo model}
    \text{(QUBO)}\quad \min_{\bx} \quad\{f(\bx) = \bx^\top \bQ \bx + c\}
\end{equation} 
where $\bx\in\bbB^n$, $\bQ \in \bbR^{n\times n}$ is a matrix of real coefficients and $c \in \bbR$ is a constant. 
Another popular expression of QUBO is 
$\min_{\bx\in\bbB^{n}} \; \{f(\bx) =  \sum_{i,j:i \neq j} q_{ij} x_i x_j + \sum_i q_{ii} x_i + c\}$.
Since $ x_i^2=x_i$ in the binary case, the two representations are interchangeable. 

When representing optimization problems as QUBO models, unless the problem has a natural QUBO form, the conventional approach is to consider an existing model with a similar form (\eg a BLP) whose constraints are then relaxed and added to the objective function as quadratic penalties \cite{glover2018}. Thus, the core of deriving QUBO formulations is to derive \emph{valid} penalty terms. We first introduce the necessary conditions for a penalty to be considered valid and then briefly discuss how these constraints have been derived in the literature.

\subsection{Valid Infeasible Penalty}
The core of deriving a QUBO formulation from a BLP formulation is to formulate penalty terms for the problem constraints.
For a constraint $h(\bx)\le 0$, a function $P(\bx)$ is called a \emph{valid infeasible penalty} (\emph{VIP}) \cite{lewis_note_2009,ayodele2022} over a search domain $\bbB^n$
if it is null when the constraint holds and otherwise is strictly positive. That is, 
\begin{eqnarray*}
&P:& \bbB^n \to \bbR_{\ge 0},\; \text{s.t.}\; P(\bx) =0 \Leftrightarrow h(\bx)\le 0, \;P(\bx) >0 \Leftrightarrow h(\bx)> 0.
\end{eqnarray*}

 A function $P'(\bx,\bs)$ is an \emph{augmented VIP} for a constraint $h(\bx)\le 0$  if it satisfies the following,
\begin{eqnarray*}
P': \bbB^{n+n'} \to \bbR_{\ge 0},\; \text{s.t.}\; P(\bx)=\min_{\bs \in \bbB^{n'}} P'(\bx,\bs) \text{ is a VIP}, n'>0 ,
\end{eqnarray*}
\noindent where $\bs$ is a vector of introduced, \emph{ancillary} binary variables required to represent the constraint by a penalty term.
For the purpose of deriving QUBO formulations, we are interested in constructing a quadratic VIP (\emph{QVIP}) in this paper.  If an augmented VIP is a quadratic on $\bx, \bs$, we will refer to it as a \emph{quadratic augmented VIP} (\emph{QAVIP}).

\subsection{Transforming Linear Constraints into Valid Infeasible Penalty}
Let $\bbB^n$ be a search domain and $E\subseteq \bbB^n$ be a feasible domain of an optimization problem under consideration; the goal is to reformulate a BLP model as a QUBO model.
The core task of this transformation is to construct a function
to penalize the constraint violation such that the optimal solution not only minimizes the objective function but also satisfies the constraints of the original problem.
The penalty function is the sum of penalty terms for the associated constraints. For equality constraint $\ba^\top\bx=b$, a corresponding penalty term can be easily formed as $\lambda (\ba^\top\bx-b)^2 $, where $\lambda>0$ is a penalty constant. The transformation of inequality constraints is non-trivial. Existing work on this can be categorized into general-purpose and special-purpose approaches.

\paragraph{General Purpose Approach}
The \emph{general purpose approach} is proposed by \citeauthor{glover_quantum_2022}~\cite{glover_quantum_2022}, where the inequalities are converted to equality constraints by adding ancillary binary variables and then applying quadratization.
%
\begin{align*}
\ba^\top\bx\le b,\; \bx\in \bbB^n  \Longrightarrow & \; \ba^\top\bx + w = b,\; \bx\in \bbB^n,\; w\in\bbR \\
\Longrightarrow&\;\ba^\top\bx+\bg^{\top}\bs= b,\; \bx\in \bbB^n,\; \bs\in \bbB^{n'} \\
\Longrightarrow &\; \lambda(\ba^\top\bx+\bg^{\top}\bs- b)^2, \bx\in \bbB^n, \bs\in \bbB^{n'},\lambda>0,
\end{align*}

\noindent where $w$ is a slack variable used to convert an inequality constraint to an equality constraint. To represent $w$ in the binary form, a set of binary variables $\bs$ is introduced (\ie $w = \bg^\top\bs$), increasing the problem dimension significantly, especially when the inequality constraint is loose. Expressing $w$ as ancillary binary variables (\eg through binary expansion) may lead to a coefficient vector $\bg$ with a very large range.
%
These two factors negatively impact QUBO performance \cite{quintero_characterization_2022}.
%
Thus, it is preferable to avoid ancillary variables. 

\paragraph{Special Purpose Approach}
For some constraints, compact penalty terms (\ie with no ancillary variables) are known \citep{glover_quantum_2022, lewis_note_2009}. Table~\ref{tab: penalties} summarizes some well-known compact penalty terms that were developed for specific applications and problems, thus, we refer to this approach as the \emph{special purpose approach}. Such approaches lead to a QUBO with fewer variables due to the absence of ancillary variables.

However, the flexibility and applicability of special-purpose approaches are limited as they cannot be applied to linear inequalities in the general form.

\begin{table}[htbp]
\centering
\centering
\begin{tabular}{ | c| c|  }
\hline
	Constraint& Corresponding	Penalty\\
\hline
	\parbox{5cm}{\begin{equation} x_i + x_j \le  1\label{C1} \end{equation}} &		\parbox{5cm}{\begin{equation} x_ix_j \label{P1} \end{equation}}\\
	\parbox{5cm}{\begin{equation} x_i + x_j \ge  1\label{C2} \end{equation}}&\parbox{5cm}{\begin{equation} 1-x_i-x_j+x_ix_j \label{P2} \end{equation}}	\\
	\parbox{5cm}{\begin{equation} x_i \le  x_j\label{C3} \end{equation}}& \parbox{5cm}{\begin{equation} x_i-x_ix_j \label{P3} \end{equation}}	\\
 \parbox{5cm}{\begin{equation} x_i + x_j+x_k \le  1\label{C4}\end{equation}}&\parbox{5cm}{\begin{equation} x_ix_j+x_ix_k+x_jx_k \label{P4} \end{equation}}\\
\hline
\end{tabular}
\caption{Known penalties for inequality constraints~\cite{glover_quantum_2022}. 
}
\label{tab: penalties}
\end{table}

 \subsection{QUBO Transformation Scheme}\label{ssec:TR0-TR4}

%
To facilitate the presentation of QUBO model derivation, we define a list of transformation procedures. In this section, we list five transformation procedures extracted from the literature. In Section \ref{sec: MLCTS}, we introduce two new transformation procedures. Lastly, we formally present the conventional BLP-to-QUBO transformation scheme (CTS) in the literature \cite{glover_quantum_2022,kochenberger_unified_2004}.


\paragraph{Transformation Procedures}
    \noindent Transformation 0 (\textbf{TR0}):  convert linear inequalities into equations by introducing a slack variable which is then replaced by a set of ancillary binary variables and then quadratize.
              \begin{equation*}
    \eqref{eq: LBP ineq constr} \quad  \Longrightarrow \quad \bA''\bx +  \bB\bs=\bb'',\;  \bs \in \bbB^{n'} \quad \Longrightarrow\quad (\bA''\bx + \bB\bs-\bb'')^2.\label{eq: LBP ineq constr2}
        \end{equation*}
    To replace a slack variable with ancillary binary variables,  one may use either discrete expansion  \cite{pisaruk_mixed_2019,verma_penalty_2022} or binary expansion  \cite{verma_penalty_2022}, which are further referred to as TR0.1 and TR0.2, respectively.\vspace*{.5em}
    
    %
    \noindent Transformation 1 (\textbf{TR1}): perform quadratization of equality constraints \eqref{eq: LBP eq constr}.
        \begin{equation*}
        \eqref{eq: LBP eq constr}  \Longrightarrow A\bx-\bb=0 \Longrightarrow (A\bx-\bb)^2.\label{eq: LBP eq constr2}
        \end{equation*}

    \noindent Transformation 2 (\textbf{TR2}): replace inequality constraints on $\bx$ by a quadratic VIP without ancillary variables. The step applies known VIPs (e.g., given in Table~\ref{tab: penalties}) or deriving penalties for specific problems (e.g., the Linear Ordering Problem \cite{lewis_note_2009}) or specific constraints (see Section~\ref{sec: MLCTS developing}). \vspace*{.5em}

    \noindent Transformation 3 (\textbf{TR3($r$)}): transform a polynomial VIP into a multilinear form \cite{boros_compact_2020} by decreasing the degrees of binary variables based on the property $x^p=x, \forall x\in \{0,1\}, \forall p\ne 0$, then factoring out a common factor $r>0$ of the coefficients.\footnote{If $r$ is not specified, notation TR3 will be used.}\vspace*{.5em}
    
    \noindent  Transformation 4 (\textbf{TR4}): polynomial reduction of cubic and other high-degree monomials in a VIP, which results in the formation of a  QAVIP \cite{boros_compact_2020,rosenberg_reduction_1975,mulero-martinez_polynomial-time_2021,freedman_energy_2005,ishikawa_transformation_2011,anthony_quadratic_2017f}. This can be done in two different ways.
    
 
\noindent 
\begin{itemize}
    \item
    The first approach (\textbf{TR4.1}) achieves the degree reduction by substitution. The degree of polynomials is reduced to 2 iteratively by making the following substitution.
    \begin{eqnarray}\label{eq: z=xy}
    z=xy
    \end{eqnarray}
    for binary variables $x,\; y$ in high-order terms, along with incorporating the Rosenberg Penalty~\cite{rosenberg_reduction_1975}
       \begin{eqnarray} \label{eq: Rosen}
        R(x,y,z)=xy-2xz-2yz+3z
    \end{eqnarray}    
    into the objective. We refer to \eqref{eq: Rosen} as the Rosenberg VIP for the rest of this paper. \vspace*{.5em}

\item 
    The second approach (denoted as \textbf{TR4.2}) achieves degree reduction by minimum selection \cite{mulero-martinez_polynomial-time_2021}. It is applied to each monomial depending on its sign.
     Let us consider a monomial $\varphi(\bx,a)=a \prod_{i\in I}x_i$, 
    where $a\in\{-1,1\}$, $I\subseteq \overline{1,n}$. Depending on the sign $a$ of the monomials in a polynomial VIP, the following formulas can be used for 
    $\varphi(\bx,a)$ quadratization.
        \begin{itemize}
        \item negative monomial \cite{freedman_energy_2005}:
         \begin{equation}\label{eq: min selection formulanegative}
          \varphi(\bx,-1)= -\min_{s_0} s_0 (S_1(I)-|I|+1), \text{ where } s_0 \in \bbB
         \end{equation}
         \item   positive monomial \cite{ishikawa_transformation_2011}: 
           \begin{equation}\label{eq: min selection formula2}
         \varphi(\bx,1)= \min_{ s_1,...,s_{n_{|I|}}}\sum_{i-1}^{n_{|I|}} s_i (c_{i,|I|} (-S_1(I)+2i)-1) + S_2(I), \text{ where }, s_1,\:\dots,\:s_{n_{|I|}}\in\bbB
         \end{equation}
          
        where 
            \begin{equation}\label{eq: S1,S2}
            \begin{aligned}
        &S_1(I) = \sum_{i\in I} x_i, S_2(I) = \sum_{i, j\in I, i<j}  x_ix_j =S_1(I) (S_1(I)-1)/ 2\\   
        &c_{i,|I|} =
        \begin{cases}
        1,\text{ if }|I|\text{ is odd and }i = n_{|I|}\\
         2, \text{ otherwise}
        \end{cases}\\
        &  n_{|I|} = \left\lfloor \frac{|I|-1}{2}\right\rfloor .
            \end{aligned}
            \end{equation}
        \end{itemize} 
\end{itemize}

\paragraph{Conventional Transformation Scheme}
The general purpose (\ie TR1) and special purpose (\ie TR2) approaches are widely used in the literature \cite{glover2018, glover_quantum_2022} to derive VIPs of associated linear constraints.  QUBO models are then obtained by combining the resulting VIPs, penalty coefficients, and the objective function. Lastly, minimization is performed over the original variables (and ancillary variables, if applicable). To distinguish it from our novel QUBO model scheme, we refer to this QUBO modeling approach as the conventional transformation scheme (CTS).

\paragraph{Compact and Augmented Formulations}

A BOP formulation with ancillary binary variables in addition to $\bx$ is referred to as an \emph{augmented BOP}. Consequently, if all augmented BOP constraints are linear, it is an \emph{augmented BLP}. 
In this paper, we assume that a BLP formulation of the problem of interest is given, and we propose a novel modelling scheme to transform a BLP formulation into a QUBO formulation.  A QUBO model with the same solution space (\ie same decision variable of $\bx$) as the corresponding BLP is referred to as a \emph{compact QUBO} formulation, and a QUBO model in the lifted solution space (\ie with ancillary binary variables) is referred to as an \emph{augmented QUBO} formulation.\\

%% file: sec3-method.tex
\section{Methodology}\label{sec: Method}
In this section, we introduce the concept of multilevel linear constraints and discuss how that can be used to derive VIPs. We then formally introduce the Multilevel Constraint Transformation Scheme (MLCTS). Lastly, using MLCTS, we derive necessary conditions for the existence of a compact QUBO formulation in terms of constraint levelness.

\subsection{Preliminaries}
We refer to the inequalities with the following form as a two-sided constraint.
 \begin{eqnarray}\label{eq: 2-side}
\underline{b}_i\le \ba_i^\top \bx \le \overline{b}_i,
\end{eqnarray}
where $\bx\in \bbB^n$, $\ba\in \bbR^n$, $\underline{b}, \overline{b}\in \bbR$, $\underline{b}\le \overline{b}$. It can also be expressed in the following discrete form.
 \begin{eqnarray}\label{eq: h(x)}
h(\bx)=\ba_i^\top \bx,\; h(\bx) \in \cH,
\end{eqnarray}
\noindent where $\cH \subset \bbR$ is a set of values that takes the function $h(\bx)$ in the closed interval  $[\underline{b}, \overline{b}]$.

\vspace*{5pt}
\begin{remark}
    If the following condition holds, then $h(\bx)\in \cH$ can be replaced by the weaker condition $h(\bx)\in \cH'=\overline{\underline{b},\overline{b}}$.
    \begin{equation}\label{eq: TR0 integer cond}
        \ba\in \bbZ^n, \:\underline{b},\overline{b} \in \bbZ.
    \end{equation}
\end{remark}

We refer to the inequalities with the following form as one-sided constraints.
\begin{align} 
    &\ba^\top \bx \le \overline{b} \label{eq: 1-side}\\
    & \ba^\top \bx \ge \underline{b} \label{eq: 1-side2}
\end{align}

\subsection{Multilevel constraints}\label{sec: MLCs}
Without loss of generality, we can assume that in BLP \eqref{eq: LBP objective}-\eqref{eq: LBP ineq constr}, inequality constraints are written in a two-sided form as \eqref{eq: 2-side}.
%
Let $\bbB^n$ be a search domain, and let a linear constraint \eqref{eq: 2-side} be imposed.
The function \eqref{eq: h(x)} takes a finite number of  values on $\bbB^n$, i.e., there exists
\begin{eqnarray} \label{eq: H}
&H=\{h_1, ..., h_K\}\subset \bbR^1, &h_1<...<h_K,
\end{eqnarray} 
such that $\bx\in \bbB^n \:\Leftrightarrow \: \exists i\in \overline{1,K}:\: h(\bx)=h_i$.
%
We introduce notations for sets of indices of nonzero entries of $\ba$ as below. Note, $K$ is bounded from above by the value $2^{|I|}$.
\begin{eqnarray}\label{eq: I,I+,I-}
I^+=\{i: a_i>0\}, I^-=\{i: a_i<0\}, I=I^+\cup I^-.
\end{eqnarray}

\begin{remark} \label{rem: 1}
Without loss of generality, we assume that $\underline{b},\overline{b}\in H$. Otherwise, the domain    $[\underline{b},\overline{b}]$ of values of $h(\bx)$ can be narrowed by setting   $\underline{b}=b_1,\overline{b}=b_k$. 

%
In addition, we assume that \eqref{eq: 2-side} is valid, \ie $\overline{b}-\underline{b}<b_K-b_1$, otherwise, the constraint is redundant.
\end{remark}
\hspace{1em}

\begin{definition}
  Inequality \eqref{eq: 2-side} is called a $k$-level linear constraint over $\bbB^n$ if 
  \begin{eqnarray}\label{eq: k-level cond}
 & \exists i: &\underline{b}=h_i,  \overline{b}=h_{i+k-1}.  
\end{eqnarray}
\end{definition}
This implies that, in the binary domain subject to \eqref{eq: 2-side}, $h(\bx)$ takes exactly $k$ consecutive values in $H$ starting with $h_i$. 
Let $H^{i,k}$ denote the set of values that $h(\bx)$ takes.
\begin{eqnarray}\label{eq: hik}
H^{i,k}=\{h_i,...,h_{i+k-1}\}.
\end{eqnarray}
 
If $\underline{b},  \overline{b}$ satisfy \eqref{eq: k-level cond}, then a $k$-level constraint  \eqref{eq: 2-side} can be written as the follows.
\begin{equation}\label{eq: k-level combinatorial}
    \exists i\in \overline{1,K},\; h(\bx)\in H^{i, k}.  
\end{equation}

	
This discrete constraint can be represented as a polynomial equation below.
\begin{eqnarray}
&\Pi^{i,k}(\bx)=\prod_{h'\in H^{i, k}} (h(\bx)-h')=0.\label{eq: 2s P}
\end{eqnarray}	
 $\Pi^{i,k}(\bx)$ has a polynomial of degree $k$.	

We refer to a $k$-level constraint that satisfies 
\begin{equation}\label{eq: 2-sided} 
h_1<\underline{b}\le \overline{b}<h_K,
\end{equation}
as two-sided; otherwise, the constraint is referred to as one-sided. 
The latter means that either $\underline{b}=h_1$ or $\overline{b}=h_K$, and the one-sided constraint is either an upper-bound constraint \eqref{eq: 1-side} 
or a lower-bound constraint \eqref{eq: 1-side2}. 
%
Let  constraint \eqref{eq: 2-side} be $k$-level, namely, \eqref{eq: k-level combinatorial} holds. 
\vspace*{5pt}
\begin{lemma}\label{lem: 1}
If $k$ is even, then the following expression is a VIP for \eqref{eq: 2-side}.
\begin{equation}
P(\bx)=\Pi^{i,k}(\bx)\label{eq: l12 P(x)}
\end{equation}
\end{lemma}

Examining if  $P(\bx)$ is a VIP for constraint \eqref{eq: 2-side} includes two steps, where we verify the following conditions, respectively.
\begin{align}
&\underline{b}\le \ba^\top \bx \le \overline{b} \Rightarrow P(\bx)=0,\; \bx \in \bbB^n \label{eq: Cond1} \\
& \ba^\top \bx \notin [\underline{b},\overline{b}] \Rightarrow P(\bx)>0,\; \bx\in \bbB^n \label{eq: Cond2} 
\end{align}


\begin{proof}
We first verify condition \eqref{eq: Cond1}. If $\bx\in \bbB^n$ satisfies \eqref{eq: 2-side}, then \eqref{eq: k-level combinatorial} holds. Consecutively, \eqref{eq: 2s P} also holds at $\bx$ (as a result of \eqref{eq: k-level combinatorial}). Thereby, condition \eqref{eq: Cond1} is satisfied.

We then examine condition \eqref{eq: Cond2}. If $\bx\in \bbB^n$ does not satisfy \eqref{eq: 2-side}, then we have either $h(\bx)>\overline{b}$ or $h(\bx)<\underline{b}$. In the former case, $P(\bx)=\prod_{h'\in H^{i, k}}(h(\bx)-h')>0$
is a product of positive factors, thus, $P(\bx)$ is positive. In the latter case, $P(\bx)$ is also positive since it is a product of an even number of negative values.
\end{proof}

\begin{corollary}\label{Cor: 1}
The following expression is a QVIP for a 2-level constraint \eqref{eq: 2-side}.
\begin{equation}
P(\bx)=(h(\bx)- h_i)(h(\bx)- h_{i+1})=(h(\bx)-\underline{b})(h(\bx)-\overline{b}).\label{eq: QVIP}
\end{equation}
\end{corollary}

\begin{lemma}\label{lem: 2}
If $\underline{b}= h_1$, then \eqref{eq: l12 P(x)} is a VIP for \eqref{eq: 2-side}.
\end{lemma}


Lemma \ref{lem: 2} states that if \eqref{eq: 2-side} is a one-sided constraint \eqref{eq: 1-side}, then $\Pi^{1,k}(\bx)$ defines a VIP for the constraint.
\begin{proof}
\eqref{eq: l12 P(x)} takes the form of $P(\bx)=\Pi^{1,k}(\bx)=\prod_{j=1}^k (h(\bx)-h_i)$.
Given $\bx\in \bbB^n$, we have either $h(\bx)\le h_k$ or $h(\bx)> h_k$. In the former case,  $P(\bx)=0$. Thereby, condition \eqref{eq: Cond1} is satisfied. In the latter case, $P(\bx)>0$ since it is a product of positive factors. Thereby, condition \eqref{eq: Cond2} is satisfied.
\end{proof}


\begin{corollary}\label{Cor: 2}
A polynomial VIP corresponding to a one-sided constraint \eqref{eq: 1-side2}
is 
\begin{eqnarray}
&P(\bx)=(-1)^k\Pi^{K-k+1,K}(\bx).\label{eq: K-k+1,K}
\end{eqnarray}
\end{corollary}

\noindent If $k$ takes an even value, formula \eqref{eq: l12 P(x)} always defines a VIP. If $k$ takes an odd value, Lemmas~\ref{lem: 1} and \ref{lem: 2} do not cover the case of two-sided constraints. In this situation, we propose increasing the degree of $\Pi^{i,k}(\bx)$ by one to form a VIP. 

\vspace*{5pt}
\begin{lemma}\label{lem: 3}
If $k$ is odd, then, for $j\in \overline{i,i+k-1}$, the following expression is a VIP for \eqref{eq: 2-side}. 
\begin{equation}
P(j,\bx)=(h(\bx)-h_{j})\Pi^{i,k}(\bx)\label{eq: 2-side,ikj}
\end{equation}
\end{lemma}

\begin{proof}
In \eqref{eq: 2-side,ikj}, $P(j,\bx)$ equals to zero at any point $\bx$ satisfying \eqref{eq: k-level combinatorial}. Thereby, condition \eqref{eq: Cond1} is satisfied. Meanwhile, for $\bx: h(\bx)>h_{i+k-1}$, the value $P(j,\bx)>0$ since it is a product of positive values. For $\bx: h(\bx)<h_{i}$, $P(j,\bx)$ is also positive since it is a product of an even number of negative factors. Thereby, condition \eqref{eq: Cond2} is satisfied.
\end{proof}

\begin{corollary}\label{Cor: 3}
The following expression is a QVIP for an equality constraint $\ba^\top \bx =h_i$.
\begin{eqnarray}
&P(i,\bx)=(h(\bx)-h_{i})^2. \label{eq: 1-leve vIP}
\end{eqnarray}

\end{corollary}
\noindent Note that this corollary yields a VIP used in TR1 (see \cite{glover_quantum_2022}). Thus, Lemma~\ref{lem: 3} is a generalization of the conventional scheme of forming QVIPs for linear equality constraints. Lemmas~\ref{lem: 1}, \ref{lem:  2}, and \ref{lem: 3} define polynomial VIPs for an arbitrary $k$-level constraint.

For a 2-level constraint on $\bx$, a QVIP can be obtained using \eqref{eq: QVIP}. For a 2-level constraint $h(\bx,\bs)\in \{\underline{b},\overline{b}\}$ involving $\bx$ and $\bs$, an QAVIP can be obtained as the follows.
\begin{eqnarray}
&P(\bx,\bs)=(h(\bx,\bs)-\underline{b})(h(\bx,\bs)-\overline{b})\label{eq: AQVIP}
\end{eqnarray}

The proofs of the validity of \eqref{eq: 1-leve vIP} and \eqref{eq: AQVIP} are straightforward and thus omitted here.


\subsection{The Multilevel Constraint Transformation Scheme}\label{sec: MLCTS}

We first introduce two new transformation procedures, and then formally define the Multilevel Constraint Transformation Scheme (MLCTS) to derive compact or augmented QUBO formulations from a binary linear program. In addition, we derive the necessary conditions for the existence of a compact QUBO reformulation.



\paragraph{New Transformation Procedures}
\begin{itemize}
    \item Transformation 5 (\textbf{TR5}): derive a (augmented) QVIP using Corollary~\ref{Cor: 1}. \\

\item Transformation 6 (\textbf{TR6}): this procedure is applied to transform equality constraints(s) with ancillary variables into 2-level constraint(s) by excluding either one or two ancillary variables, which are referred to as \textbf{TR6.1} and \textbf{TR6.2}, respectively. We discuss the details of these two variants in Section \ref{ssec: TR6}.
\end{itemize}

\subsubsection{ MLCTS outline}

The MLCTS is formally defined in Algorithm \ref{alg:MLCTS}. Starting with a BLP model in the form of \eqref{eq: LBP Model}, MLCTS derives (augmented) QVIPs for associated equality and inequality constraints to form a QUBO model. 

MLCTS starts with equality constraints and derives QVIPS using TR1 (see line 1). It then iterates through each inequality constraint (line 2 - 17). For each inequality constraint, we first investigate its form by refining the bounds and finding the levelness (line 2 and 3). Using the information, we determine the appropriate ways to derive a VIP (\ie $P'(\bx)$) for the associated inequality constraint (line 5 - 9). In line 10, the polynomial degree of the VIP is reduced by TR3, if applicable. For VIP with a polynomial degree of three or higher, MLCTS offers two options for degree reduction to obtain an augmented QVIP (line 11 - 15). Lastly, all QVIPs and augmented QVIPs are combined with the objective $f(\bx)$ to form a penalty function (line 18).




Note, in Option 2 of polynomial degree reduction (line 13 - 15), depending on the choice of TR0 procedure (TR0.1 or TR0.2), the corresponding TR6 procedures will be performed (TR6.1 or TR6.2). We further discuss this process in Section \ref{ssec: TR6}.

As shown above, MLCTS offers two options for polynomial degree reduction when $\kappa'\ge 3 $. For the purpose of notation, we will introduce the notation MLCTS($\cdot)$, where $(\cdot)$ indicates the variants of degree reduction in line 11-16. For instance, the notation MLCTS(TR4.1) means that this step is performed using TR4.1 (line 12) and the notation MLCTS (TR0+TR6) means that this step is performed with Option 2 (line 13-15).



\begin{algorithm}[hbt]
\caption{Multilevel Constraint Transformation Scheme}\label{alg:MLCTS}
\begin{algorithmic}[1]
\Require A BLP model in the form of \eqref{eq: LBP Model}.
 \Ensure  A QUBO model in the form of \eqref{eq: qubo model}.

 \State Perform TR1 to get QVIPs for associated equality constraints.

 \For {each inequality constraint}
    \State  Refine lower and upper bounds  $\underline{b},\overline{b}$,  if applicable.
    \State Find the levelness $k$ of the constraint.
    \If {the constraint is an upper bound one-sided constraint, \ie $\ba^\top\bx\le\overline{b}$}
          Derive $P(\bx)$  using Lemma~\ref{lem: 2}
    \ElsIf{the constraint is a lower bound one-sided constraint, \ie $\ba^\top\bx\ge\underline{b}$}
          Derive $P(\bx)$ using Corollary~\ref{Cor: 2}.
    \ElsIf{the constraint is a two-sided \textit{k-level} constraint and $k$ is even}
          Derive $P(\bx)$ using  Lemma~\ref{lem: 1}.
    \Else{ Derive $P(\bx)$ using  Lemma~\ref{lem: 3}.}
    \EndIf  \Comment{$P(\bx)$ is a polynomial VIP with degree $\kappa\le k+1$}

 \State Perform TR3 on $P(\bx)$ to derive a VIP $P'(\bx)$ with a degree $\kappa'\le \min\{n, |I|, \kappa\}$.

\If{$\kappa'\ge 3 $ } do one of the following
\State \textbf{Option 1}: Perform TR4 on $P'(\bx)$ and obtain  an augmented QVIP  $P''(\bx,\bs)$ 
\State \textbf{Option 2}: Perform TR0 on the ineq. constraint, obtain eq. constraint(s).
\State \phantom{\textbf{Option 2}:} 
            Perform TR6 to transform eq. constraint(s) into 2-level inequality constraint(s)
\State \phantom{\textbf{Option 2}:} Apply TR5 to derive corresponding augmented QVIP(s).
\EndIf
 \EndFor

\State Combine all the QVIPs and augmented QVIPs in a linear combination with positive coefficients and add the objective $f(\bx$).

\State \Return
\end{algorithmic}
\end{algorithm}

\subsubsection{Sufficient Conditions for the Existence of a compact QUBO Reformulation}\label{ssec: suf cond}
To derive a compact QUBO formulation, there are two key factors: the constraint levelness and the number of variables involved in the constraints. In this section, we derive the sufficient conditions when a compact QUBO formulation can be derived from a given BLP model.

Let $P(\bx)$ and $P'(\bx)$ be VIPs for the $k$-level constraint \eqref{eq: 2-side} obtained from line 5-9 and line 10 
of Algorithm~\ref{alg:MLCTS}, respectively. Then the degree $\kappa$ of $P(\bx)$ is
\begin{equation} \label{eq: existence-of-cqubo-degree-k}
\kappa= \left\{
\begin{aligned}
&k+1,\: \text{ if }P'(\bx)\text{ is found by }\eqref{eq: 2-side,ikj}\\
&k,\; \text{otherwise}
\end{aligned}\right.
\end{equation}


\begin{lemma}\label{lem: 6}
The degree $\kappa'$ of a polynomial VIP $P'(\bx)$ for the $k$-level constraint \eqref{eq: 2-side}  obtained by the MLCTS satisfies  $\kappa'\le \min \{ \kappa, |I|\}$.
\end{lemma}
\begin{proof}
$P'(\bx)$  is a polynomial of degree $\kappa$ with monomial terms on at most $|I|$ variables. When TR3 is performed on the polynomial $P(\bx)$, its terms are converted into multilinear monomials on the same variables. Hence, the degree of the resulting VIP $P'(\bx)$  does not exceed $|I|$.
\end{proof}

Lemma~\ref{lem: 6} allows us to derive the sufficient conditions for the existence of a QUBO formulation of a BLP \eqref{eq: LBP objective}-\eqref{eq: LBP ineq constr}.
Consider the $i^\text{th}$ constraint in \eqref{eq: LBP ineq constr}, 
$h_i(\bx)=\ba_i^{'\top} \bx \le b_i', i\in \overline{1,m'}$,
and let $I_i\subseteq \overline{1,n}$ be a set indices of nonzero coefficients of a vector $\ba'_i$.  $\kappa_i$ is a levelness of the $i\text{th}$ constraint in \eqref{eq: LBP ineq constr} such that 
\begin{equation}\label{eq: existence-of-cqubo-degree-k'}
\kappa_i=
\begin{cases}
k_i+1\text{ if a VIP for }i\text{th constraint in } \eqref{eq: LBP ineq constr}\text{ is found by }\eqref{eq: 2-side,ikj}\\
k_i\text{, otherwise}.
\end{cases}
\end{equation}

\begin{theorem}\label{th: 1}
If the following condition is satisfied,  then the BLP \eqref{eq: LBP objective}-\eqref{eq: LBP ineq constr} allows a compact QUBO formulation.
\begin{eqnarray}\label{eq: KN <=2}
KN=\underset{i\in \overline{1,m'}}\max \min \{\kappa_i,|I_i|\}\le 2,
\end{eqnarray}
where $KN$ is an upper bound on the degree of the polynomial VIPs found from line 10 of Algorithm~\ref{alg:MLCTS}.
\end{theorem}
\begin{proof}
\eqref{eq: KN <=2} implies that every inequality constraint in \eqref{eq: LBP ineq constr}  allows quadratization using TR3 in accordance to Lemma~\ref{lem: 6}.  Therefore, TR1 can be applied directly for quadratization of equality constraints \eqref{eq: LBP eq constr} and the formation of QVIPs. Combining all the VIPs together in a linear combination with positive coefficients, which is then incorporated into the objective, we get a QUBO equivalent to the BLP \eqref{eq: LBP objective}-\eqref{eq: LBP ineq constr} if adjusted penalty constants are chosen.
\end{proof}

\begin{corollary}\label{cor: 6.0}
If each inequality constraint in \eqref{eq: LBP ineq constr} has levelness at most two, then the BLP \eqref{eq: LBP objective}-\eqref{eq: LBP ineq constr} allows a compact QUBO formulation.
\end{corollary}
\begin{proof}
    If $k_i=1$, then $1\le\kappa_i\le 2$; if $k_i=2$, then $\kappa_i= 2$. Respectively, $KN\le \max\: \{2,...,2\}=2$. That is, Theorem~\ref{th: 1} is satisfied.
\end{proof}
\begin{corollary}\label{cor: 6}
If each inequality constraint in \eqref{eq: LBP ineq constr} involves at most two  variables, then the BLP \eqref{eq: LBP objective}-\eqref{eq: LBP ineq constr} allows a compact QUBO formulation.
\end{corollary}
\begin{proof}
    For each $i$, $|I_i|\le 2$ hence $KN\le\underset{i\in \overline{1,m'}}\max \: |I_i|\le 2$. That is, Theorem \ref{th: 1} is satisfied.
\end{proof}


\subsubsection{MLCTS Illustration}\label{ssec: MLCTS il1}
We demonstrate the MLCTS on a synthetic optimization problem represented as the following BLP model, where $x_1,x_2,x_3\in \{0,1\}$ are binary variables.
\begin{subequations}\label{eq: ex1}
\begin{align}
\text{(BLP1)}\quad \minimize_{x_1,x_2,x_3} &\quad x_1+x_2+2x_3\\
        \subjto &\quad 0\le x_1+2x_2-x_3\le 2 \label{eq:ex1-1}\\
                &\quad  1\le 2x_1+2x_2-x_3\le 2\label{eq:ex1-2}\\
                &\quad  3x_1-2x_3\ge 1 \label{eq:ex1-3} 
\end{align}
\end{subequations}

\paragraph{Approach 1: Conventional Approach.} In the conventional approach, we convert all inequality constraints to equality constraints by introducing ancillary binary variables as below. 
\begin{align*}
    0\le x_1+2x_2-x_3\le 2 \quad \Longrightarrow &\quad x_1+2x_2-x_3\in\{0,1,2\}\\
    \underset{TR0}\Longrightarrow & \quad   \left\{\begin{aligned}
                                                    &x_1+2x_2-x_3+s_1+2s_2=2,\; s_1,s_2\in\{0,1\},\\&
                                                    s_1+s_2\le 1, \\
                                             \end{aligned} \right.\\    
    1\le 2x_1+2x_2-x_3\le 2 \quad \underset{TR0}\Longrightarrow &\quad \left\{\begin{aligned}
                                                        &2x_1+2x_2-x_3 + s_3=2,\; s_3\in\{0,1\},\\
                                                        & 3x_1-2x_3\ge 1 
                                        \end{aligned}\right.\\
    \Longrightarrow & \quad 0\le 3x_1-2x_3+2\le 3 \\
    \underset{TR0}\Longrightarrow & \quad 3x_1-2x_3+2+s_4+2s_5= 3, s_4,s_5\in\{0,1\}
\end{align*}

We then impose all equality constraints into the objective function as penalties by performing TR1 and obtaining an augmented QUBO model below. 
\begin{align*}
\text{(CTS-QUBO1)}\quad \minimize_{\bx,\: \bs} &\quad x_1+x_2 +2x_3 \\
                                                             &\quad +\lambda_1 (x_1+2x_2-x_3+s_1+2s_2-2)^2\\ 
                                                             &\quad +\lambda'_1 s_1s_2\\
                                                             &\quad + \lambda_2 (2x_1+2x_2-x_3 + s_3-2)^2 \\
                                                             &\quad +\lambda_3 (3x_1-2x_3+s_4+2s_5- 1)^2,
\end{align*}

where $\lambda_1,\lambda'_1,\lambda_2,\lambda_3>0$ are justified penalty constants, and $\bx\in \bbB^3,\bs\in \bbB^5$.

\paragraph{Approach 2: MLCTS.}  We first derive a set of VIPs for each constraint in the BLP model. 
For constraint (\ref{eq:ex1-1}), we know that $x_1+2x_2-x_3 \in H=\{-1,0,1,2,3\},\;K=5,\; \underline{b}=0,\overline{b}=2,\; i=1, k =3$. Thus, this is a \textit{3-level} two-sided constraint with $H^{1,3} = \{0,1,2\}$. According to Lemma~\ref{lem: 3}, we have the following VIP.
    \begin{align*}
        &\quad P_1(\bx) = (x_1+2x_2-x_3)(x_1+2x_2-x_3-1)^2(x_1+2x_2-x_3-2)\\
        \underset{TR3(12)}\Longrightarrow &\quad P_1'(\bx)=x_1x_2-x_1x_3-x_2x_3+x_3
    \end{align*}
\noindent For constraint (\ref{eq:ex1-2}), we know that $2x_1+2x_2-x_3\in H=\{-1,0,1,2,3,4\},\;K=6, \underline{b} = 1,\; \overline{b}=2,\; i=k=2$. Thus, this is a \textit{2-level} constraint with $H^{2,2} =\{1,2\}$. According to Corollary \ref{Cor: 1}, we have the following VIP. 
    \begin{align*}
        &\quad P_2(\bx) = (2x_1+2x_2-x_3-1)(2x_1+2x_2-x_3-2) \\
        \underset{TR3(2)}\Longrightarrow &\quad P_2'(\bx)=  4x_1x_2-2x_1x_3-2x_2x_3-x_1-x_2+2x_3+1
    \end{align*}
\noindent For constraint (\ref{eq:ex1-3}), we know that $3x_1-2x_3\in H=\{-2,0,1,3\},\;K=4,\;\underline{b}=-2,\;\overline{b}=1,\: i=0,\:k=3$. Thus, this is a \textit{3-level} one-sided constraint of type \eqref{eq: 1-side} with $H^{0,3}=\{-2,0,1\}$. According to Lemma~\ref{lem: 2}, we have the following VIP.
    \begin{align*}
        &\quad P_3(\bx) = (3x_1-2x_3+2)(3x_1-2x_3)(3x_1-2x_3-1) \\
        \underset{TR3(30)}\Longrightarrow &\quad P_3'(\bx) = -x_1x_2+x_1
    \end{align*}



\noindent All VIPs are quadratic. We combine them with the objective and obtain a QUBO model below.
\begin{align*}
\text{(MLCTS-QUBO1)}\quad \minimize_{\bx} &\quad x_1+x_2 +2x_3 \\
                                             &\quad +\lambda_1 (x_1x_2-x_1x_3-x_2x_3+x_3) \\
                                             &\quad +\lambda_2 (4x_1x_2-2x_1x_3-2x_2x_3-x_1-x_2+2x_3+1)\\
                                             &\quad +\lambda_3 (-x_1x_2+x_1),
\end{align*}
where $\lambda_1,\lambda_2,\lambda_3>0$ are justified penalty constants. 

Let us check if conditions of Theorem~\ref{th: 1} hold for BLP1: $(k_1,|I_1|)=(3,3)$, $(k_2,|I_2|)=(2,3)$, $(k_3,|I_3|)=(3,2)$, $KN=\max \: \{3,2,2\}=3$. The sufficient condition \eqref{eq: KN <=2} on the existence of QUBO reformulation does not hold. At the same time, such a reformulation was found by applying MLCTS. Thus, condition \eqref{eq: KN <=2} is sufficient but not necessary.
The dimensions of MLCTS-QUBO1 and CTS-QUBO1 are $3$ and $8$, respectively.


%% file: sec4-mlcts-constraint.tex
\section{MLCTS: Further development}\label{sec: MLCTS developing}
In this section, we further discuss the use of MLCTS, particularly how to derive VIPs for 2-level and high-level constraints. 
More illustration examples can be found in Appendix \ref{sec:MLCTS-illustration}.

\subsection{Exploring 2-level constraints} \label{ssec: families}
Two-level constraints play a special role in QUBO modelling since, as can be seen from the formula \eqref{eq: QVIP}, they induce QVIPs. 
Additionally, as stated in Corollary \ref{cor: 6}, BLPs with the constraint levelness of at most two allow a compact QUBO reformulation.
One must note that all constraints presented in Table~\ref{tab: penalties} are 2-level, explaining the QVIPs listed in the table. 

We refer to linear constraints with coefficient values of $0$, $1$, and $-1$  as \emph{regular constraints}.  Under such a restriction, \eqref{eq: I,I+,I-} becomes 
$I^+=\{i\in \overline{1,n}: a_i=1\}, I^-=\{i\in \overline{1,n}: a_i=-1\}, I=\{i\in \overline{1,n}: a_i\ne 0\}$.
While, function $h(\bx)$ becomes 
$h(\bx)=\sum_{i\in I^+} x_i- \sum_{i\in I^-} x_i$
and it takes all integer values in the range $[z^{\min}, z^{\max}]$, where $z^{\min}=\underset{\bx\in \bbB^n}\min\: h(\bx)=-|I^-|,\; z^{\max}=\underset{\bx\in \bbB^n}\max \:h(\bx)=|I^+|.$
Lastly, we also have $H=\{\:h_1,\: ...,\: h_{|I|+1}\:\}\:=\:\{-|I^-|,\:-|I^-|+1,\:...,\:|I^+|-1,\:|I^+|\}$.
%
Consider a regular constraint \eqref{eq: 2-side} $\underline{b}\le \sum_{i\in I^+} x_i- \sum_{i\in I^-} x_i \le \overline{b}$,
its levelness is $k=\overline{b}-\underline{b}+1$. We may derive a QVIP as follows.



\vspace*{5pt}
\begin{theorem} \label{th: 01}
If $j\in [-|I^-|,|I^+|-1]$, then
\begin{equation}\label{eq: QVIP reg}
\begin{aligned}
P^j(\bx)=&\sum_{i<i'; i,i'\in I^+} x_ix_{i'}+\sum_{i<i'; i,i'\in I^-} x_ix_{i'}-\sum_{i\in I^+,i'\in I^- } x_ix_{i'}\\
& -j\sum_{i\in I^+} x_i+(j+1)\sum_{i\in I^-} x_i+\frac{j(j+1)}{2}
\end{aligned}
\end{equation}
is a QVIP associated with a 2-level regular constraint
\begin{eqnarray}\label{eq: 2-sidereg1}
j\le \sum_{i\in I^+} x_i- \sum_{i\in I^-}x_i \le j+1.
\end{eqnarray}
\end{theorem}

\begin{proof}
By Cor.~\ref{Cor: 1}, we know that 
\begin{eqnarray*}
P(\bx)=(\sum_{i\in I^+} x_i- \sum_{i\in I^-}x_i-j)(\sum_{i\in I^+} x_i- \sum_{i\in I^-}x_i-j-1)
\end{eqnarray*}
is a QVIP for \eqref{eq: 2-sidereg1}. We simplify this expression and transform it into a  multilinear polynomial as below.

\begin{align*}
P(\bx)=&\; (\sum_{i\in I^+} x_i- \sum_{i\in I^-}-j)^2-(\sum_{i\in I^+} x_i- \sum_{i\in I^-}x_i-j)\\
=&\; (\sum_{i\in I^+} x_i)^2+ (\sum_{i\in I^-}x_i)^2+j^2-2 \sum_{i\in I^+,i'\in I^-} x_ix_{i'}- 2 j\sum_{i\in I^+} x_i \\
&\; + 2 j\sum_{i\in I^-} x_i - \sum_{i\in I^+} x_i+ \sum_{i\in I^-}x_i+j\\
\underset{TR3(2)}\Longrightarrow &\;\frac{1}{2}\cdot \Big(\:(2\sum_{i<i'; i,i'\in I^+}x_ix_{i'}+\sum_{i\in I^+} x_i) +(2\sum_{i<i'; i,i'\in I^-}x_ix_{i'}+\sum_{i\in I^-} x_i)+j^2-\\
&\;-2 \sum_{i\in I^+,i'\in I^-} x_ix_{i'}- 2 j\sum_{i\in I^+} x_i + 2 j\sum_{i\in I^-} x_i-\sum_{i\in I^+} x_i+ \sum_{i\in I^-}x_i+j\:\Big)\\
=&\;\sum_{i<i'; i,i'\in I^+} x_ix_{i'}+\sum_{i<i'; i,i'\in I^-} x_ix_{i'}-\sum_{i\in I^+,i\in I^- } x_ix_{i'}\\
&\; - j\sum_{i\in I^+} x_i+(j+1)\sum_{i\in I^-} x_i+\frac{j(j+1)}{2} = P^j(\bx).
\end{align*}


\end{proof}
 
Furthermore, we derive QVIPs for the case $I=I^+$, that is when there are no negative terms are present in a regular constraint.

\vspace*{5pt}
\begin{corollary}\label{Cor: 4}
For every $j\in \overline{1, |I|-1}$, 
\begin{eqnarray}\label{eq: QVIP reg01}
P^{j}(\bx)=\sum_{i<i'; i,i'\in I} x_ix_{i'}-j\sum_{i\in I} x_i+\frac{j(j+1)}{2}
\end{eqnarray}
is a QVIP associated with a 2-level constraint with $0,1$ coefficients 
\begin{eqnarray}\label{eq: 2-sidereg01}
j\le \sum_{i\in I} x_i \le j+1.
\end{eqnarray}
\end{corollary}
This corollary can be used to derive QVIPs for some commonly used constraints, such as a (binary) conflict constraint \cite{pferschy_approximation_2017}
\begin{eqnarray}\label{eq: conflict}
\sum_{i\in I} x_i \le 1
\end{eqnarray}
or a (binary) forcing constraint \cite{pferschy_approximation_2017}
\begin{eqnarray}\label{eq: forcing}
\sum_{i\in I} x_i \ge |I|-1.
\end{eqnarray}

\begin{corollary}\label{Cor: 5}
$\forall I \subseteq \overline{1,n}$,  
\begin{eqnarray}\label{eq: QVIP conflict}
P^c(\bx)=\sum_{i<i'; i,i'\in I} x_ix_{i'}
\end{eqnarray}
is a QVIP corresponding to a conflict constraint \eqref{eq: conflict}, and 
\begin{eqnarray}\label{eq: QVIP forcing}
P^f(\bx)=\sum_{i<i'; i,i'\in I} x_ix_{i'}+|I|\sum_{i\in I} x_i+ \frac{|I|(|I|-1)}{2}
\end{eqnarray}
is a QVIP corresponding to a forcing constraint \eqref{eq: forcing}.
\end{corollary}

\begin{proof}
$P^c(\bx)$ is derived from \eqref{eq: QVIP reg01} by substitution $j=1$, while $P^f(\bx)$ is obtained from this formula using substitution $j=|I|-1$.
\end{proof}

Furthermore, we provide examples to demonstrate how MLCTS can be used to derive other QVIPs listed in Table~\ref{tab: penalties} and many others.

\begin{itemize}
    \item The known VIPs \eqref{P1} and \eqref{P4} can be derived from \eqref{eq: QVIP conflict}  by choosing $I=\{i,j\}$ and $I=\{i,j,k\}$, respectively.
    \item For deriving the QVIP \eqref{P3} for constraint \eqref{C3}, we utilize  \eqref{eq: forcing} with $I=\{i,j\}$.
    \item The penalty \eqref{P2} can be found by \eqref{eq: QVIP reg01} by selecting $j=-1$,  $I^+=\{i\}$, $I^-=\{j\}$.
\end{itemize}




\subsection{Exploring High-Level Constraints}\label{ssec: TR6}
As noted above, Transformation 6 (TR6) is applied to reduce high-level constraints to 2-level constraints. After TR0 (\ie converting inequality to equality), TR6 decreases the number of ancillary variables and constructs QVIPs or augmented QVIPs. We offer two variants: Transformation 6.1 (\textbf{TR6.1}) following TR0.1 
and  Transformation 6.2 (\textbf{TR6.2}) following  TR0.2.

We consider a $k$-level constraint as \eqref{eq: 2-side} with $k\ge 3$.
The first step is the normalization of the function $h(\bx)$ in order to make its two smallest values equal  $0$ and $1$, \ie 

\begin{eqnarray}\label{eq: h0(x)}
h\left( \bx \right)\Longrightarrow {{h}^{0}}\left( \bx \right)=\frac{h\left( \bx \right)-{{h}_{i}}}{{{h}_{i+1}}-{{h}_{i}}}
\end{eqnarray}


Consequently, \eqref{eq: k-level combinatorial} can be rewritten as
\begin{subequations}
\begin{align}
    &{{h}^{0}}\left( \bx \right)\in {{H}^{0,j,k}}={{\left\{ h_{j}^{0} \right\}}_{j\in \overline{1,k}}} \label{eq:h0-transform}\\
    &h_{1}^{0}=0,h_{2}^{0}=1\label{eq: h0 norm}\\
    &h_{j}^{0}=\frac{h_j-{{h}_{i}}}{{{h}_{i+1}}-{{h}_{i}}}, j\in \overline{3,k}\label{H0ik}
\end{align}
\end{subequations}
where $h^0_1$ or $h^0_2$ are the first and second value of  $h^0(\bx)$  indicated in \eqref{eq:h0-transform}.

\subsubsection{Transformation 6 Variant 1 (TR6.1)}
After performing TR0.1,
we obtain the following equality constraint,
\begin{equation}\label{eq: TR0.1.0'}
  {h}^{0}\left( \bx \right)- \sum_{i=1}^kh^0_is_i=0
 \end{equation}
 Note, $s_i,\forall i$ are binary ancillary variables which are constrained by an additional one-hot constraint (\ie $ \sum_{i=1}^ks_i=1$).
%
The two ancillary variables which are denoted as $s_1,\: s_2$
can be eliminated.
From \eqref{eq: h0 norm}, we can rewrite equation \eqref{eq: TR0.1.0'} to eliminate $s_1$.
 \begin{eqnarray}\label{eq: TR0.1.0''}
&  {h}^{0}\left( \bx \right)- s_2- \sum_{i=3}^kh^0_is_i=0,
 \end{eqnarray}
 We then sum the one-hot constraint 
 and \eqref{eq: TR0.1.0''} to eliminate $s_2$.
  \begin{eqnarray}\label{eq: TR0.1.1''}
&  {h}^{0}\left( \bx \right)+s_1+ \sum_{i=3}^k(1-h^0_i)s_i=1,
 \end{eqnarray}
 We may derive 2-level inequalities from \eqref{eq: TR0.1.0''} and \eqref{eq: TR0.1.1''} as below.
\begin{align}\label{eq: phi1,2}
    &\; \left\{\begin{aligned}
      & {{\varphi }_{1}}\left( \bx,\bs \right)={{h}^{0}}\left( \bx \right)-\sum\limits_{i=3}^{k}{h_{i}^{0}{{s}_{i}}}={{s}_{2}}\\ 
     & {{\varphi }_{2}}\left( \bx,\bs \right)={{h}^{0}}\left( \bx \right)+\sum\limits_{i=3}^{k}{\left( 1-h_{i}^{0} \right){{s}_{i}}}=1-{{s}_{1}}
    \end{aligned}\right. \nonumber \\
    \Longrightarrow&\;\left\{\begin{aligned}
      & {{\varphi }_{1}}\left( \bx,\bs \right)\in \left\{ 0,1 \right\}\\ 
     & {{\varphi }_{2}}\left( \bx,\bs \right)\in \left\{ 0,1 \right\} 
    \end{aligned}\right. \nonumber \\
    \Longrightarrow&\;\left\{\begin{aligned}
      &  
      0\le {{\varphi }_{1}}\left( \bx,\bs \right) \le 1, 0\le {{\varphi }_{2}}\left( \bx,\bs \right) \le 1.
    \end{aligned}\right.
\end{align}


Given \eqref{eq: phi1,2} 
we perform TR5 
to derive augmented QVIPs (see \eqref{eq: AQVIP}).
\begin{eqnarray} \label{eq: P1,2}
{{\varphi }_{j}}\left( \bx,\bs \right)\Longrightarrow {{P}^{j}}\left( \bx,\bs \right)={{\varphi }_{j}}\left( \bx,\bs \right)\left( {{\varphi }_{j}}\left( \bx,\bs \right)-1 \right),\:j=1,2.
\end{eqnarray}

The final augmented QVIP for the high-level constraint \eqref{eq: 2-side} is a linear combination of these two with  coefficients $\lambda_1,\lambda_2$
 \begin{equation}
{{P}_{\lambda }}\left( \bx,\bs \right)=\sum\limits_{j=1}^{2}{{{\lambda }_{j}}{{P}^{j}}\left( \bx,\bs \right)}, \; \lambda_1,\lambda_2\in \bbR_{>0}.			\label{eq: (1.1)}
\end{equation}


We further illustrate this procedure in the examples with $k=3$ in Appendix \ref{sec:appendix-higher-level}.

\subsubsection{Transformation 6 Variant 2 (TR6.2)}

\noindent If \eqref{eq: TR0 integer cond} holds, TR0.2 
is applicable, and we obtain the following
\begin{align}
\underline{b}\le h(\bx)\le \overline{b} \Longrightarrow &\; h(\bx)+ w= \overline{b},\; w\in[\underline{b},  \overline{b}]\nonumber  \\
\Longrightarrow &\;
     h(\bx)+ w= \overline{b},\; w=\underline{b}+\sum_{i=0}^{n'-1} 2^i s_i, \label{eq: tr 0.2}\\
    &\text{where, } s_i\in\bbB,\; i\in \overline{0,n'-1},\; n'=\left\lceil \log_2 (\overline{b}-\underline{b}+1) \right\rceil. \nonumber
\end{align}

\noindent Since $h\left( \bx \right)\in \overline{ i,i+k-1}$, we obtain that,
\begin{align}
 &\;{{h}^{0}}(\bx)\in \overline{0,k-1 },\; \text{where},\: {{h}^{0}}(\bx)=h\left( \bx \right)-i \nonumber\\ 
 \Longrightarrow&\; {{h}^{0}}\left( \bx \right)+w=k-1\nonumber\\
 &\;w\in \overline{0,k-1 }. \label{eq: is}
\end{align}

In a standard binary expansion, the integer slack variable $w$ is replaced by a linear combination of ancillary binary variables in accordance to \eqref{eq: tr 0.2}, that is, 
\begin{eqnarray*}
w=\sum_{i=0}^{n'-1} 2^i s_i, s_0,...,s_{n'-1}\in \bbB,\;\text{where } n' \text{ is given by \eqref{eq: tr 0.2}.}
\end{eqnarray*}

In general, this condition represents a relaxation, $w\in \left[ 0,2^{n'}-1 \right]$, instead of \eqref{eq: is}. Respectively, if a 2-sided constraint (\ie \eqref{eq: 2-sided}) is under consideration and the levelness $k$ is not an integer power of two, TR0.2 needs modification like below.

\vspace*{5pt}
\begin{proposition}\label{prop: 2}
    \begin{align}
    & w=\sum\limits_{j=0}^{n'-2}{{{2}^{j}}{{s}_{j}}}+a{{s}_{n'-1}}=k-1, s_0,...,s_{n'-1}\in \bbB, \label{eq: TR0.2n} \\
&\text{where } \;
a=k-1-\sum\limits_{j=0}^{n'-2}{{{2}^{j}}}, n' \text{ is given by  \eqref{eq: tr 0.2},} \label{eq: a}
\end{align} 
is an equivalent binary representation of \eqref{eq: is}.
\end{proposition} 
This variant of TR0.2 is further referred to as \textbf{TR0.2n}. 
%
Now, we separate $s_0$ in \eqref{eq: TR0.2n} and introduce an ancillary function ${{\varphi }_{3}}\left( \bx,\bs \right)$:
\begin{align*}
&\; {{h}_{0}}\left( \bx \right)+w={{h}_{0}}\left( \bx \right)+{{s}_{0}}+\sum\limits_{j=1}^{n'-2}{{{2}^{j}}{{s}_{j}}}+a{{s}_{n'-1}}=k-1\\
\Longrightarrow &\; {{\varphi }_{3}}\left( \bx,\bs \right)={{h}_{0}}\left( \bx \right)+\sum\limits_{j=1}^{n'-2}{{{2}^{j}}{{s}_{j}}}+a{{s}_{n'-1}}=k-1-{{s}_{0}} \nonumber\\
\Longrightarrow&\;{{\varphi }_{3}}\left( \bx,\bs \right)\in \left\{ k-2,k-1 \right\} \nonumber\\
\Longrightarrow&\; k-2\le {{\varphi }_{3}}\left( \bx,\bs \right) \le k-1.
\end{align*}


Similar to \eqref{eq: phi1,2}, we derive to a 2-level constraint with fewer ancillary variables by applying TR5 to ${{\varphi }_{3}}\left( \bx,\bs \right)$, and induce an AQVIP $ {{P}_{{3}}}\left( \bx,\bs \right)$ from \eqref{eq: AQVIP}:

\begin{align} 
&{{P}_{{3}}}\left( \bx,\bs \right)=\left( {{\varphi }_{3}}\left( \bx,\bs \right)-k+1 \right)\left( {{\varphi }_{3}}\left( \bx,\bs \right)-k+2 \right), \label{eq: P3 1} \\
&\text{where},\;{{\varphi }_{3}}\left( \bx,\bs \right)={{h}_{0}}\left( \bx \right)+\sum\limits_{j=1}^{n'-2}{{{2}^{j}}{{s}_{j}}}+a{{s}_{n'-1}} \label{eq: phi3}
\end{align}


We further illustrate this procedure in the examples with $k=3$ in Appendix \ref{sec:appendix-higher-level}.

%% file: sec5-mlcts-cop.tex
\section{MLCTS for well-known problems}\label{sec: applied problems}

In this section, we apply the MLCTS to four well-known combinatorial optimization problems: Maximum 2-Satisfiability (Max2SAT), the Linear Ordering Problem (LOP), the Community Detection (CD) problem, and the Maximum Independent Set Problem (MIS).
\subsection{ Maximum 2-Satisfiability Problem}
 In the Maximum 2-Satisfiability problem (Max2SAT) \cite{glover_quantum_2022}, a Conjunctive Normal Formula (CNF) is 
$\text{CNF}=\wedge_{i=1}^m C_i$, where a clause $C_i$ involves two literals $l_{i1}, l_{i2}$, i.e. $C_i=l_{i1} \vee l_{i2}$, $i\in \overline{1,m}$. Each literal is a propositional variable from a set $\{x_1,...,x_n\}$ or its negation; it is a function on $\bx$ (\ie $l_{ij}=l_{ij}(\bx)$). 
The goal is to find a vector $\bx=(x_1,...,x_n)$ that maximizes the number of satisfied clauses.

Let $x_j \in\bbB$ denote the binary variables where $j\in \overline{1,n}$ is the corresponding index. $C_i \in\bbB$ denotes the binary variables representing the clauses where $i\in \overline{1,m}$ is the corresponding index of clauses.

The Max2SAT can be formulated as the following BLP model.
%
%
\begin{align}
\text{(Max2SAT.BLP1)}\quad\minimize_{\bx,\: \bC} &\quad \{f(\bC)=\sum_{i=1}^m \neg C_i=m- \sum_{i=1}^m C_i\}\label{eq: Max2SAT obj}\\
\subjto&\quad l_{i1}+l_{i2}\ge C_i, i\in \overline{1,m} \label{eq: Max2SAT C1}\\
&\quad l_{i1}\le C_i, i\in \overline{1,m} \label{eq: Max2SAT C2}\\
&\quad l_{i2}\le C_i, i\in \overline{1,m} \label{eq: Max2SAT C3}\\
&\quad \bx \in \bbB^n, \bC \in \bbB^m,l_{i1}=l_{i1}(\bx),l_{i2}=l_{i2}(\bx), i\in \overline{1,m},\label{eq: Max2SAT C4}
\end{align}

where $\bC=(C_1,...,C_m)$ is a vector of clauses.

The objective is to maximize the number of satisfied clauses, although it is written in the form of minimization. We express the relation between clauses and literals 
through a set of linear constraints \eqref{eq: Max2SAT C1} - \eqref{eq: Max2SAT C3}.

Note, \eqref{eq: Max2SAT C1} is a 3-level constraint as $h_i=l_{i1}+l_{i2}- C_i\in \{0,1,2\}$. However, a feasible finite point configuration described by a set of constraints \eqref{eq: Max2SAT C1}-\eqref{eq: Max2SAT C4} for a fixed $i$  includes four points $(l_{i1},l_{i2}, C_i)\in \{(0,0,0), (0,1,1), (1,0,1), (1,1,1)\}$, where $h_i$ varies from $0$ to $1$ only. Thus, \eqref{eq: Max2SAT C1} can be refined to 
\begin{eqnarray}
& C_i\le l_{i1}+l_{i2}\le C_i+1,\; i\in \overline{1,m}. \label{eq: Max2SAT C1.1}
\end{eqnarray}
Thus, a refined BLP (referred to as \textbf{Max2SAT.BLP2}) is given by  the objective \eqref{eq: Max2SAT obj} and constraints \eqref{eq: Max2SAT C2}-\eqref{eq: Max2SAT C1.1}. 
Notably, all these linear constraints are regular 2-level. In particular, \eqref{eq: Max2SAT C2} and \eqref{eq: Max2SAT C3} are known constraints in the form of \eqref{C3}, which has a corresponding known VIP \eqref{P3}. Applying formula \eqref{eq: QVIP reg} with parameters $j=0$, $I^+=\{(i,1),(i,2)\}$, and $I^-=\{i\}$, we get a QVIP $P_i^1=1-l_{i1}-l_{i2}+2C_i-l_{i1}l_{i2}-l_{i1}C_i-l_{i2}C_i$ for the constraint \eqref{eq: Max2SAT C1.1}. Combining it with \eqref{C3} we get an augmented QUBO 

\begin{align*}
    \text{(Max2SAT.QUBO1)}\:\minimize_{\bx,\: \bC }\:  F_{\lambda}(\bx,\bC)\label{eq: Max2SAT Q1} =&\; m- \sum_{i=1}^m  C_i+\lambda_1 \sum_{i=1}^m P_i^1 +\lambda_2 \sum_{i=1}^m (l_{i1}-l_{i1}C_i)\\
    &\;+ \lambda_3\sum_{i=1}^m (l_{i2}-l_{i2}C_i),
\end{align*}


\noindent where $\bx\in \bbB^n$, $\bC\in \bbB^m$.  
Lastly, $\mathbf{\lambda}=(\lambda_1,\lambda_2,\lambda_3)\in \bbR^3_{>0}$ is a vector of adjusted penalty constants.
 Max2SAT.QUBO1 may also be expressed in $\bx$. To do that, we rewrite Max2SAT.QUBO1 in terms of propositional variables by substitution with the below expressions of literals.

\begin{equation*}\label{eq: Ci(x)}
C_i=l_{i1}\vee l_{i2}=
\left\{
    \begin{aligned}
    &x_{j_{i1}}\vee x_{j_{i2}},i\in I^{++}\\
    &x_{j_{i1}}\vee \neg x_{j_{i2}},i\in I^{+-}\\
    &\neg x_{j_{i1}}\vee \neg x_{j_{i2}},i\in I^{--}
    \end{aligned}\right.,
\end{equation*}

\noindent where sets $I^{++},\: I^{+-},\: I^{--}$ form a partition of a set $\overline{1,m}$ depending on the number of positive literals ($2,1,0$ for $I^{++},\: I^{+-},\: I^{--}$, respectively), $\{j_{i1},j_{i2}\}\subseteq \overline{1,n}$ is a set of indices of propositional variables in a clause $C_i$, $i\in \overline{1,m}$. 
$j_{i1},\:j_{i2}\in \overline{1,n}$ are the indices of a propositional variable present in $C_{i1}$ and $C_{i2}$, respectively. For instance if $C_5=x_3\:\vee\: \neg x_7$, then $i=5$, $j_{i1}=3$, $j_{i2}=7$.


We then take a closer look at the expression of the clause $C_i$. Note that a clause $C_i$ depends on two propositional variables $x_{j_{i1}}$, $ x_{j_{i2}}$. Thus, a negation $\neg C_i$ can be replaced by a quadratic VIP that depends on these variables only,  according to  Cor.~\ref{cor: 6}. 
In addition, $\neg C_i=0$ if and only if $l_{i1}+l_{i2}\ge 1$, and $\neg C_i=1$ if and only if $l_{i1}+l_{i2}=0$. The inequality $l_{i1}+l_{i2}\ge 1$  has the form of \eqref{C2}. Correspondingly, adapting \eqref{P2} for constructing a VIP we get $P_i=1-l_{i1}-l_{i2}+l_{i1}l_{i2}$. Depending on the type of the clause, we may have three expressions of $P_i$ as a function of $\bx$ as follows.
\begin{align}
P_i=&\;1-l_{i1}-l_{i2}+l_{i1}l_{i2} \nonumber\\
=
&\;\left\{\:\begin{aligned}
P_i^{++}=&\:1-x_{j_{i1}}-x_{j_{i2}}+x_{j_{i1}}x_{j_{i2}},i\in I^{++}\\
P_i^{+-}=&\:x_{j_{i2}}-x_{j_{i1}}x_{j_{i2}},i\in I^{+-}\\
P_i^{--}=&\:x_{j_{i1}}x_{j_{i2}},i\in I^{--}.
\end{aligned}\right. \label{eq: Pi(x)}
\end{align}


 \noindent It easy to see that $P_i=\neg C_i$ for binary $l_{i1},l_{i2}$, and, we obtain the following compact QUBO formulation.

\begin{equation*}
    \begin{aligned}
    \text{(Max2SAT.QUBO2)}\;\minimize_{\bx}\: F(\bx) =&\; \sum_{i=1}^m  \neg C_i = \sum_{i=1}^m  P_i\\
    =&\; \sum_{i\in S^{++}}P_i^{++}+\sum_{i\in S^{+-}}P_i^{+-}+\sum_{i\in S^{--}}P_i^{--},
    \end{aligned}
\end{equation*}
where $P_i^{++},P_i^{+-},P_i^{--}$ are given by \eqref{eq: Pi(x)}.

The dimension of Max2SAT.QUBO1 (an augmented QUBO model) is $n+m$, but the dimension of Max2SAT.QUBO2 (a compact QUBO model) is only $n$.
Since $m\le n(n-1)/2$,
the ratio of the dimensions may vary from
$1+1/n$ to $(n+1)/2$.
Notably, unlike the Max2SAT.QUBO1, Max2SAT.QUBO2 does not have penalty parameters. This is a very important and desirable characteristic of QUBO formulation, as tuning penalty parameters, which might be time-consuming, is crucial to the effectiveness of solving QUBO models \cite{quintero_characterization_2022}. 

\subsection{Linear Ordering Problem}

Given an $n\times n$ weight matrix $\bW= (w_{ij})_{
i,j\in \overline{1,n}}$, 
the objective of the Linear Ordering Problem (LOP) is to determine the permutation $\pi$ of its columns (and rows) such that the sum of the weights in the upper triangular part of the resulting matrix is maximized \cite{laguna_intensification_1999}. Within this context, the permutation $\pi$ provides the ordering of both columns and rows.


We started with a BLP model for LOP presented in the literature~\cite{grotschel_cutting_1984}. 
Let $i,\:j,\:k \in V= \{1,\dots,n\}$ denote the set of indices of columns (and rows) for $\bW$.
$x_{ij} \in \bbB$ denotes the decision variable representing the ordering. $x_{ij}=1$ if item $i$ precedes item $j$, otherwise, $x_{ij}=0$. The BLP model is presented below.
\begin{align}
    \text{(LOP.BLP)}\quad f^*=\max_{X} &\quad \{ f(X)=\sum_{i<j}w_{ij}x_{ij}+\sum_{j<i} w_{ij}(1-x_{ji})\}\label{eq: LOP obj}\\
    \subjto &\quad 0\le x_{ij}+x_{jk}-x_{ik}\le 1, \; i,j,k\in V,\; i<j<k\label{eq: LOPC1}\\
    &\quad x_{ij}\in \{0,1\},\; i,j \in V,\; i<j, \label{eq: LOP LOPC2}
\end{align}
where $X=(x_{ij})_{i,j\in V, i<j}$ is an upper triangular binary matrix of dimension $n$.
Since we only consider the upper triangular binary matrix, the constraint of $x_{ij} + x_{ji} = 1$ is not needed as it is incorporated in the second term of the objective (\ie $\sum_{j<i} w_{ij}(1-x_{ji})$).
Constraint \eqref{eq: LOPC1} is a triangle inequality that ensures that $x$ encodes a linear order. Notably, \eqref{eq: LOPC1}  is a set of 2-level regular constraints \eqref{eq: 2-sidereg1} for $j=0$, $I^+=\{(i,j),(j,k)\},I^-=\{(i,k)\}$. Utilizing \eqref{eq: QVIP reg} for each triple $i,j,k$, we get a QVIP $P_{i,j,k}$ corresponding to a constraint involving $h_{i,j,k}= x_{ij}+x_{jk}-x_{ik}$.\ignore{$\in\{0,1\}$.} Thus,
\begin{equation}
\;0\le h_{i,j,k}\le 1 \;\Longrightarrow \; P_{i,j,k} =x_{ik}+x_{ij}x_{jk}-x_{ij}x_{ik}-x_{jk}x_{ik},\; 1\le i < j < k \le n. \label{eq: LOP Pijk}
\end{equation}

These VIPs are presented in the literature~\cite{lewis_note_2009}. Using a truth table, the authors justify that $P_{i,j,k}$ is a VIP for \eqref{eq: LOPC1}. However, the derivation of this penalty term was not provided. In contrast, MLCTS yields an analytic way to derive \eqref{eq: LOP Pijk} and allows us to fill this gap. 
Combining the objective \eqref{eq: LOP obj} with the VIPs, we get a QUBO formulation of LOP in the  form of a minimization problem:
\begin{equation*}
    \begin{aligned}
        \text{(LOP.QUBO)} &\quad f^{**}=-f^* =F_{\lambda}(X^*)=\min_{X}\; \{ F_{\lambda}(X)\}\\
        \text{where, }&\; F_{\lambda}(X)=-\sum_{i,j\in V:\:i<j}w_{ij}x_{ij}-\sum_{i,j\in V:\:j<i} w_{ij}(1-x_{ji})\\
        &\; \phantom{ F_{\lambda}(X')=} +\lambda \sum_{i,j,k\in V:\:i < j < k }(x_{ik}+x_{ij}x_{jk}-x_{ij}x_{ik}-x_{jk}x_{ik}).\label{eq: LOP Q}
    \end{aligned}
\end{equation*}

\noindent $X$ is an binary upper-triangular matrix of order $n$, i.e. $X=(x_{i,j})_{i,j\in V, i<j}$ is binary, $\lambda\in \bbR_{>0}$ is an adjusted penalty constant.



\subsection{Community Detection}\label{ssec:CDP-QUBO}

 Given an undirected graph $G(V,E)$ on $n$ vertices and $m$ edges such as $V=\{1,...,n\}$.
Let  $A=(a_{uv})_{u,v\in V}$ be an adjacency matrix of $G$.
The general Community Detection Problem (CDP) is typically formulated as a modularity maximization problem, that is, to find a partition $\bc^*$ of $V$ into clusters such that the following modularity function $Q(\bc|G)$ is maximized.

\begin{eqnarray}\label{eq: modularity 1}
Q(\bc^*|G)=\max_{\bc}\;\{Q(\bc|G)=\frac{1}{2m}\sum_{u,v\in V}(a_{uv}-\frac{d_ud_v}{2m})\varphi(u,v)\},
\end{eqnarray}
where $\bC$ is a set of $V$-partitions, $\bc\in \bC$ is a  $V$-partition, $d_u=\sum_{v\in V}a_{uv}$ is the degree of vertex $u$, $c_u\in \bc$ is a cluster to which vertex $u$ is assigned, $\varphi(u,v)=1$ if $c_u=c_v$, otherwise, $\varphi(u,v)=0$.


Let $X=(x_{uv})_{ u, v \in V}$ denote the matrix of binary decision variables. $x_{uv}=1$ if $u$ and $v$ are in different clusters, otherwise, $x_{uv} = 0$. A known BLP model of the CDP is as follows \cite{agarwal_modularity-maximizing_2008}.
\begin{align}
\text{(CDP.BLP)}\quad f(X^*)=\max_{X}&\quad \{f(X)=\sum_{u,v\in V}a_{uv}(1-x_{uv})\}\label{eq: modularity 1.1} \\
 \subjto &\quad x_{uw} \le x_{uv} + x_{vw},\; \{u,v,w\}\subseteq V. \label{eq: CD transitivity}
\end{align}


The objective is to maximize the modularity function. Constraint \eqref{eq: CD transitivity} is a transitivity constraint. We may rewrite \eqref{eq: CD transitivity} as below.
\begin{equation} \label{eq: CD transitivity2}
h_{u,v,w}= x_{uv} + x_{vw}-x_{uw},\; 0\le h_{u,v,w} \le 2,\; \{u,v,w\}\subseteq V
\end{equation}

It is easy to see that $h_{u,v,w}\in \{0,1,2\}$ for each triple $u,v,w$. Hence, the transitivity constraints are $3$-level one-sided constraints in the form of  \eqref{eq: 1-side2}. Applying Cor.~\ref{Cor: 2} to \eqref{eq: CD transitivity2} gives us a set of cubic VIPs as follows.
\begin{align*}
&\; h_{u,v,w}\in \{0,1,2\},\; \{u,v,w\}\subseteq V\\
\Longrightarrow &\; P_{u,v,w}=\;-h_{u,v,w}(h_{u,v,w}-1)(h_{u,v,w}-2),\; \{u,v,w\}\subseteq V\\
\underset{TR3(6)}\Longrightarrow &\;
P'_{u,v,w}=\;x_{uv} x_{vw}x_{uw}-x_{uv} x_{uw}- x_{vw}x_{uw}+x_{uw},\; \{u,v,w\}\subseteq V.
\end{align*}

\noindent The validity of VIP $P'_{u,v,w}$ is verified in Table~\ref{tab: 0.2} in the Appendix.
%
%
We then conduct quadratization of $P'_{u,v,w}$ using TR4.1.
\begin{align}
&\;y_{u,v,w}=x_{uv}x_{vw},\; \{u,v,w\}\subseteq V\label{eq: yuvw}\\
\Longrightarrow &\;P''_{u,v,w}=y_{uvw}x_{uw}-x_{uv} x_{uw}- x_{vw}x_{uw}+x_{uw},\; \{u,v,w\}\subseteq V. \label{eq: CD VIPs2}\nonumber
\end{align}

\noindent Replacing each constraint in \eqref{eq: yuvw} with the corresponding Rosenberg penalty gives us 
\begin{equation*}
R_{u,v,w}=x_{uv}x_{vw}-2 x_{uv}y_{uvw}-2 x_{vw}y_{uvw}+3y_{uvw},\; \{u,v,w\}\subseteq V . \label{eq: Ruvw}
\end{equation*} 

Lastly, an augmented QUBO formulation of \eqref{eq: modularity 1} is obtained.
\begin{align}
\text{(CDP.QUBO1)}\quad  F_{\lambda_1,\lambda_2}(X^*,Y^*)&\:=\:\min_{X,Y}\; F_{\lambda_1,\lambda_2}(X,Y)\label{eq: CDP: F^*}\\
F_{\lambda_1,\lambda_2}(X,Y)&\:=\:-f(X)+\sum_{\{u,v,w\}\subseteq V }(\lambda_1 P''_{u,v,w}+\lambda_2 R_{u,v,w})
\label{eq: CDP: F}
\end{align} 

$f(X)$ is given by \eqref{eq: modularity 1.1}, and $Y=(y_{u,v,w})_{\{u,v,w\}\subseteq V}$ is a three-dimensional array of ancillary binary variables, $\lambda_1,\lambda_2\in \bbR_{>0}$ are adjusted penalty constants.

%
$P''_{u,v,w}$ is a penalty term defined in a lifted space which depends on not only $X$ but also $Y$.
In the lifted space, it is no longer a VIP. 
When using
$\lambda_1=1$, $\lambda_2=2$, we get a valid penalty term $P'''_{u,v,w}=P''_{u,v,w}+2 R_{u,v,w}\ge 0$. This is because it returns a strictly positive value at unfeasible points 
$(x_{uv},x_{vw},x_{uw},y_{uvw})\in\{(0,0,1,0),(0,0,1,1)\}$ (see Table~\ref{tab: 0.3} in the Appendix). Moreover, $\min_{y_{uvw}\in \{0,1\}}P'''_{u,v,w}$ is a VIP for \eqref{eq: CD transitivity2} coinciding with $P'_{u,v,w}$ on binary domain (see Table~\ref{tab: 0.2} in the Appendix). 

Hence, $P'''_{u,v,w}$ is an augmented VIP for the constraint and can be used in QUBO formulations.

Utilizing $P'''_{u,v,w}$ in \eqref{eq: CDP: F}, we come to another augmented QUBO of CDP that does not involve any penalty parameters: 

\begin{align*}
&\text{(CDP.QUBO2)}\quad F_{1,2}(X^*,Y^*)\:=\:\min_{X,Y}\; F_{1,2}(X,Y);\\
&\text{where,} \;F_{1,2}(X,Y)=-f(X)+\sum _{\{u,v,w\}\subseteq V }P'''_{u,v,w}\\
&\phantom{\text{where,}} \;=-f(X)+\sum _{\{u,v,w\}\subseteq V}(\: -x_{uv}x_{uw}-x_{vw}x_{uw}+2x_{uv}x_{vw}-4x_{uv}y_{uvw}\\
&\phantom{\text{where,} \;=-f(X)+\sum _{\{u,v,w\}\subseteq V}(\:} -4x_{vw}y_{uvw}+x_{uw}y_{uvw}+x_{uw}+6y_{uvw}\:)
\end{align*} 



\subsection{Maximum Independent Set Problem}\label{ssec:MIS}

The Maximum Independent Set Problem (MIS) \cite{garey_computers_1979} is a well-known strongly NP-hard problem. We apply both MLCTS and the standard transformation scheme to derive QUBO formulations. In addition,  in the next section, we conduct proof-of-concept numerical experiments to evaluate the performance of the different QUBO formulations on benchmark instances from the literature~\cite{noauthor_dimacs_nodate}. 

\subsubsection{MIS Formulations}
Given an undirected graph $G(V,E)$ on $n$ vertices and $m$ edges such as $V=\{1,...,n\}$, the goal of the MIS is to find the largest possible subset of vertices $S\:\subseteq\: V$ such that no two vertices in $S$ are adjacent. The cardinality of $S$ is called the independence number, denoted as $\alpha(G)$.
Here, we consider a BLP model of MIS given in the literature \cite{pardalos_maximum_1994}. 
Let $\bx=(x_v)_{v\in V}\in \mathbb{B}_n$ denote a vector of the binary decision variables.
$x_v = 1$ if $v\in S$, otherwise, $x_v = 0$.

\begin{align}
    \text{(MIS.BLP)}\quad \alpha(G)= f(\bx^*)=\max_{\bx} &\quad \{f(\bx)=\sum_{v\in V} x_v\} \label{eq: MIS obj}\\
    \subjto &\quad   x_v + x_u\le 1, \{v, u\} \in E. \label{eq: MIS constraints}
\end{align}

The objective is to maximize the number of vertices in the independent set. Constraint \eqref{eq: MIS constraints} prevents selection of adjacent vertices. 

The MIS.BLP contains inequality constraints only. Applying the CTS to its transformation requires introducing a binary slack variable for each constraint. TR0 introduces a binary slack variable $y_{vu}$ for each pair of $\{v,u\}\in E $, and consequently, 
\eqref{eq: MIS constraints} becomes the following.
\begin{align}
&x_v+x_u+ y_{vu}=1,\; \{v, u\} \in E \label{eq: MIS equality constraints} \\
&y_{vu}\in \{0,1\},\; \{v, u\}\in E.\label{eq: MIS y constraints}
\end{align}

As a result, we obtain an augmented BLP model with the objective \eqref{eq: MIS obj} and constraints \eqref{eq: MIS equality constraints}, \eqref{eq: MIS y constraints}. This model is further referred to as the MIS.ABLP (for \emph{augmented} BLP).

\paragraph{Deriving QUBO with CTS} 
Using TR1 on constraint \eqref{eq: MIS equality constraints}, for every $\{v,u\}\in E$, we derive
 \begin{align*}
&\; x_v+x_u+ y_{vu}=1 \\
\underset{TR1}\Longrightarrow &\; P_{v,u} = ( x_v+x_u+ y_{vu}-1)^2\\
\underset{TR3(1)}\Longrightarrow&\; P_{v,u}=2x_vx_u+ 2x_uy_{vu}+ 2x_vy_{vu}-x_v-x_u- y_{vu}+1, \text{ where } y_{vu} \text{ is binary}
 \end{align*}
 

Incorporating $P_{v,u}$ into the objective, we obtain the following augmented QUBO with the dimension of $m+n$.
\begin{align*}
\text{(MIS.QUBO1)}\; -\alpha(G)\:=\: &\;F_\lambda(\bx^*, \by^*)\\
\:=\: &\;\min_{\bx,\: \by }\;\big\{F_\lambda(\bx, y)=-f(\bx)+\lambda \sum_{\{v,u\}\in E} P'_{v,u}\big\}\\
\:=\: &\;\min_{\bx,\: \by}\; \big\{-\sum_{v\in V} x_v+\lambda \sum_{\{v,u\}\in E} (2x_vx_u+ 2x_uy_{vu}\\
&\;+ 2x_vy_{vu}-x_v-x_u- y_{vu}+1)\big\}
\end{align*}

The vector $\by$ is formed from the ancillary binary variables associated with edges of $G$ ($\by=(y_{vu})_{\{v,u\}\in E} \in\bbB^m$). $\lambda\in \bbR_{>0}$ is a penalty constant. 
 
\paragraph{Deriving QUBO with MLCTS} \label{ssec:exp-MIS}
Notably, \eqref{eq: MIS constraints} are regular constraints with a levelness of two. Therefore we may apply the MLCTS directly to perform quadratization without introducing ancillary variables.

Using MLCTS, we apply TR5 on constraints \eqref{eq: MIS constraints}. In accordance with \eqref{eq: conflict} and \eqref{eq: QVIP conflict}, conflict constraints \eqref{eq: MIS constraints}  induce a QVIP $P^c_{v,u}=x_vx_u$  for every edge $\{v,u\}\in E$. Respectively, the summation of the QVIPs in the penalty function yields the following compact QUBO, where $\bx=\{x_v\}_{v\in V}\in \mathbb{B}_n$.
\begin{align*}
\text{(MIS.QUBO2)}\quad -\alpha(G)=F'(\bx^*)\:=\:&\;\min_{\bx}\;\{F'(\bx)=-f(\bx)+\lambda \sum_{\{v,u\}\in E} P^c_{v,u}\}\\
\:=\:&\;\min_{\bx}\; \{-\sum_{v\in V} x_v+\lambda \sum_{\{v,u\}\in E} x_vx_u\},
\end{align*}

\noindent where $\lambda \in \mathbb{R}_{>0}$  is a penalty constant. Note this model is a compact QUBO model as it does not introduce new variables. This model was presented in the literature \cite{pardalos_maximum_1994}, 
however MLCTS provides a novel alternative approach for its derivation. 

To summarize, we derive an augmented QUBO model using CTS and a compact QUBO model using MLCTS with dimensions $n+m$ and $n$, respectively. 
For MIS with dense graphs, the MIS.QUBO1 model could be significantly larger than the MIS.QUBO2. 
In the following section, we conduct a proof-of-concept experiment to computationally evaluate the performance of these two different QUBO models.

\section{Numerical Experiments}\label{sec: MIS experiment results}
We first introduce the experiment setup and then present the complete experiment results and some discussions.


\subsection{Experiment setup}

\paragraph{Benchmark Instances}
The computational experiment was conducted on the benchmark instances presented in the lierature \cite{noauthor_dimacs_nodate}.
In the dataset, there are twenty-nine graphs of the order $n\in [64,4096]$ and the size $m\in [192, 504451]$ (see Table~\ref{tab:MIS_data_par}). 

\begin{table}[htbp]
\renewcommand*{\arraystretch}{1.1}
\centering
\begin{tabular}{|>{\centering}p{.5in}>{\centering}p{1in}>{\centering\arraybackslash}p{.8in}|} \hline
 $n$ & $m$ & \#  instances  \\ \hline
	64&	[192,543]&	3\\
	128&	[512,5173]&	5\\
	256&	[1312,17183]&	5\\
	512&	[3264,54895]&	5\\
	1024&	[18944,169162]&	5\\
	2048&	[18944,504451]&	5\\
	4096&	[184320,184320]&	1\\
\hline
Total&	&	29\\
\hline
\end{tabular}
    \caption{Summary of benchmark instances.}
    \label{tab:MIS_data_par}
\end{table}

\paragraph{QUBO Solvers}
For each benchmark instance, both MIS.QUBO2 and MIS.QUBO1 models are solved on Gurobi and Digital Annealer. 

Gurobi is a well-known commercial solver, and it effectively solves QUBO as BLP through a general branch-and-bound approach.
%
We use the default configuration and set the time limit to 300s. All feasible solutions obtained by Gurobi are recorded and later processed after the experiments are completed. Thus, post-processing time is not considered in the experiment runtime.

For a QUBO solver, we use the third generation of Fujitsu Digital Annealer (DA3)\cite{DA3}. The DA3 is a hybrid system of hardware and software for encoding and solving a QUBO problem using a parallel tempering algorithm. 
The algorithm takes advantage of hardware to perform the parallel evaluation of the objective function of every one-flip neighbor of the current assignment. It has been used in several studies in solving QUBO formulations of combinatorial optimization problems \cite{tran2016, Cohen20, Zhang22, Senderovich22a}.  For details of the hardware and software aspects of DA3, we refer the readers to the literature~\cite{DA3,oshiyama_benchmark_2022}. DA3 is a heuristic solver with no optimality guarantee. We solve each QUBO model on DA3 for 5 and 10 seconds and report the best-found solution.  DA-related experiments were conducted
on the Digital Annealer environment prepared for the research use. The Digital Annealer is accessible only through a remote server and communication time, waiting and post-processing times are not considered in our result summary.

\subsubsection{Computational Results}

Note that our objective here is to compare our two QUBO models to investigate the impact of MLCTS on solving performance. Given the substantial differences in the solvers (e.g., dedicated hardware, exact vs. heuristic), a fair, direct comparison of solver performance is beyond our scope.

\begin{table}[htbp]
\centering
\begin{subtable}[h]{\textwidth}
\centering
\renewcommand*{\arraystretch}{1.1}
\begin{tabular}{|>{\centering}p{.3in}|>{\centering}p{.15in}|>{\centering}p{1.1in}>{\centering}p{1.1in}|>{\centering}p{.758in}>{\centering\arraybackslash}p{.75in}|}\hline
 \multirow{2}{*}{$n$} & \multirow{2}{*}{$k$}         & \multicolumn{2}{c|}{ Gurobi (\# ins. in 5s/10s/300s/MRT)}& \multicolumn{2}{c|}{DA3 (\# ins. in 5s/10s)} \\ 
   &  &  MIS.QUBO2       & MIS.QUBO1        &MIS.QUBO2       & MIS.QUBO1                  \\ \hline
64&	3&	3/3/3/\textbf{0.02}&	3/3/3/0.2&	3/3&	3/3\\
128&	5&	\textbf{5/5/5/0.05}&	3/3/3/1.47&	\textbf{5/5}&	1/2\\
256&	5&	\textbf{5/5/5/0.18}&	2/2/3/12.04&	\textbf{5/5}&	0/0\\
512&	5&	\textbf{3/4/5/14.37}&	0/1/3/104.3&	\textbf{5/5}&	0/0\\
1024&	5&	\textbf{2/2/3/102.29}&	0/0/2/132.74&	\textbf{5/5}&	0/0\\
2048&	5&	0/0/\textbf{1/246.87}&	0/0/0/300&	\textbf{5/4}&	--/--\\
4096&	1&	0/0/0/300&	0/0/0/300&	0/0&	--/--\\
\hline
\multicolumn{2}{c|}{Total}&	\textbf{18/19/22/94.82}&	8/9/14/121.53&	\textbf{28/27}&	4/5\\
\hline
\end{tabular}
    \caption{Performance analysis in finding an optimal solution. Each entry presents the number of MIS instances that Gurobi or DA3 finds an optimal solution within the timelimit of 5s, 10s or 300s (Gurobi only). For Gurobi, the mean relative runtime (MRT) is also reported.}
    \label{tab:comparison_alternative1}
\end{subtable}
\begin{subtable}[h]{\textwidth}
\centering
\renewcommand*{\arraystretch}{1.1}
\begin{tabular}{|>{\centering}p{.3in}|>{\centering}p{.15in}|>{\centering}p{1.1in}>{\centering}p{1.1in}|>{\centering}p{.758in}>{\centering\arraybackslash}p{.75in}|}\hline
 \multirow{2}{*}{$n$} & \multirow{2}{*}{$k$}         & \multicolumn{2}{c|}{ Gurobi (\# ins. in 5s/10s/300s)}& \multicolumn{2}{|c|}{DA3 (\# ins. in 5s/10s)} \\ 
   &  &   MIS.QUBO2       & MIS.QUBO1       & MIS.QUBO2       & MIS.QUBO1                 \\ \hline
64&	3&	3/3/3&	3/3/3&	3/3&	3/3\\
128&	5&	5/5/5&	5/5/5&	5/\textbf{5}&	5/4\\
256&	5&	5/5/5&	5/5/5&	\textbf{5/5}&	2/3\\
512&	5&	5/5/5&	5/5/5&	\textbf{5/5}&	0/0\\
1024&	5&	\textbf{5/5}/5&	3/4/5&	\textbf{5/5}&	0/0\\
2048&	5&	\textbf{5/5}/5&	0/1/5&	\textbf{5/5}&	--/--\\
4096&	1&	\textbf{1/1}/1&	0/0/1&	\textbf{1/1}&	--/--\\
\hline
\multicolumn{2}{|c|}{Total}&	\textbf{29/29}/29&	21/23/29&	\textbf{29/29}&	10/10\\
\hline
\end{tabular}
    \caption{Performance analysis in finding a feasible solution. Each entry presents the number of MIS instances that Gurobi or DA3 finds a feasible solution within the timelimit of 5s, 10s or 300s (Gurobi only).}
    \label{tab:comparison_alternative2}
\end{subtable}
\begin{subtable}[h]{\textwidth}
\centering
\renewcommand*{\arraystretch}{1.1}
\begin{tabular}{|>{\centering}p{.3in}|>{\centering}p{.15in}|>{\centering}p{1.1in}>{\centering}p{1.1in}|>{\centering}p{.75in}>{\centering\arraybackslash}p{.75in}|}\hline
 \multirow{2}{*}{$n$} & \multirow{2}{*}{$k$}         & \multicolumn{2}{c|}{ Gurobi (opt. gap \% in 5s/10s/300s)}& \multicolumn{2}{c|}{DA3 (opt. gap \% in 5s/10s)} \\ 
   &  &   MIS.QUBO2       & MIS.QUBO1       & MIS.QUBO2       & MIS.QUBO1                   \\ \hline
64&	    3&	    0/0/0              &	0/0/0&	0/0&	0/0\\
128&	5&	\textbf{0/0/0}         &	40/40/40&	\textbf{0/0}&	9.3/4.5\\
256&	5&	\textbf{0/0/0}         &	43/43/40&	\textbf{0/0}&	12/17\\
512&	5&	\textbf{3.7/1.9/0}     & 	49/47/40&	\textbf{0/0}&	--/--\\
1024&	5&	\textbf{1.0/7.8/3.6}   &	69/56/45&	\textbf{0/0}&	--/--\\
$2048$&	5&	\textbf{13/13/7.7}     &	--/27/12&	\textbf{0/0.1}&	--/--\\
$4096$&	1&	\textbf{22/22/15}      &	--/--/18&	\textbf{5.8/6.7}&	--/--\\
\hline
\multicolumn{2}{|c|}{Average}&	\textbf{5.4/4.7/2.5}   &	44/38/31&	\textbf{0.2/0.3}&	8.1/8.5\\
\hline
\end{tabular}    
\caption{Performance analysis in solution optimality gap. Each entry presents the average optimality gap of the solutions found by Gurobi and DA3 within the timelimit of 5s, 10s or 300s (Gurobi only). Let $z^*$ denote the optimal objective value, and $\overline{z}$ denote the objective value of the best-found solution. The optimality gap is computed as  $\text{opt. gap} = 100\%\cdot(z^*-\overline{z})/z^*$.}

    \label{tab:comparison_alternative3}
\end{subtable}
\caption{Summary of MIS experimental results.}
\label{tab:MIS_res}
\end{table}

\paragraph{Evaluation Metrics}
To compare the performance of the two QUBO models, we consider three metrics: the number of instances for which an optimal solution is found (but not necessarily proved), the number of instances for which a feasible solution is found, and the optimality gap with respect to an optimal solution.  For DA, we report these statistics after 5s and 10s. For Gurobi, we report these statistics after 5s, 10s, and 300s. In addition, we also report the mean runtime for Gurobi. All computational results are summarized in Table \ref{tab:MIS_res}.

Note, ``-'' in Table \ref{tab:MIS_res} means not applicable, and it occurs under two scenarios. First, the solver (Gurobi or DA3) finds no feasible solutions, and consequently, we cannot compute the mean optimality gap (see Gurobi MIS.QUBO1 results for $n=2048, 4096$ and DA3 MIS.QUBO1 results for $n=512, 1024$). Secondly, the solver failed to complete the experiment. The second scenario only occurs to DA3, which is caused by a memory error on the DA3 solver (see DA3 MIS.QUBO1 results for $n\ge 2048$ in Table \ref{tab:comparison_alternative1} - \ref{tab:comparison_alternative3}).


\paragraph{Comparison of MIS.QUBO2 and MIS.QUBO1}

As expected based on model size, MIS.QUBO2 outperforms MIS.QUBO1 for both solvers. Both Gurobi and DA3 are able to find feasible solutions for MIS.QUBO2 within 5 seconds. The performance of Gurobi and DA3 decline when solving MIS.QUBO2, especially as DA3 only finds feasible solutions for 10 out 29 instances. In terms of the solution quality, when solving MIS.QUBO2, Gurobi and DA3 find 18 and 28 optimal solutions, respectively. While with the MIS.QUBO1, Gurobi and DA3 only find an optimal solution for 8 and 4 instances within 5 seconds. A similar difference in performance can be observed after 10 seconds (and 300s for Gurobi). For those instances where optimal solutions were not found, the average optimality gap of MIS.QUBO2 model is significantly smaller than MIS.QUBO1 on both Gurobi and DA3. Lastly, we observe the mean runtime for Gurobi to solve MIS.QUBO2 is significantly shorter than MIS.QUBO1, especially for small dimensions (\ie $n\le 512$). 

The results confirm our expectations and also the key contributions of this paper, that is, a more compact QUBO model is expected to outperform augmented QUBO models. Notably, such performance superiority is observed on different optimization solvers, further highlighting the effectiveness of MLCTS and the significance of this work.

%% file: sec6-discussion.tex
\section{Discussion}\label{sec: discussion}
We demonstrate that by utilizing the concept of constraint levelness, one may effectively derive more compact valid infeasibility penalties and form more compact QUBO models. We propose a novel approach, the \emph{Multilevel Constraint Transformation Scheme} (MLCTS), to reformulate an arbitrary BLP into a QUBO.  

All the work presented above focused on BLP. One natural but not necessarily easy extension would be to generalize MLCTS to derive QUBO models from an arbitrary BOP model. For the rest of this section, we discuss some ideas on this extension and illustrate them through a small example. 

\paragraph{Generalizing MLCTS to the BOP-to-QUBO reformulation}


In Section \ref{sec: MLCs}, we focused on linear constraints (\ie \eqref{eq: LBP ineq constr} and \eqref{eq: LBP eq constr}). By replacing $\ba^\top \bx$ by $h(\bx)$ (\eg $\underline{b}\le \ba^\top\bx \le \overline{b} \Longrightarrow \underline{b}\le h(\bx) \le \overline{b},\: h(\bx) = \ba^\top \bx $), we derived and proved a set of Lemmas and Corollaries in Section \ref{sec: MLCs} which can be used for deriving VIPs. It is easy to see that, given a $h(\bx)$ in any polynomial form, all those Lemmas and Corollaries are still valid, suggesting that MLCTS is applicable to arbitrary BOP models in the form of \eqref{eq: BOP objective}-\eqref{eq: BOP ineq constr}.  

However, the sufficient conditions for the existence of a compact QUBO formulation are slightly different in this case. More specifically, in Eq. \eqref{eq: existence-of-cqubo-degree-k} and \eqref{eq: existence-of-cqubo-degree-k'}, 
we replace the formula $\kappa =k+1$ and $\kappa_i =k_i+1$ to $\kappa =k+d$ and $\kappa_i =k_i+d_i$, respectively. $d$ is the degree of the polynomial $h(\bx)$,  $d_i$ is the degree of the polynomial $h_i(\bx)$, $i\in \overline{1,m'}$. 

For demonstration, we focus on the following quadratic equality constraint and illustrate two approaches to derive the corresponding QVIP using MLCTS.


\begin{equation}\label{eq: yij=xixj}
h(\bx,\by)=x_ix_j- y_{ij}=0.
\end{equation}

\noindent \textit{Approach 1.} \noindent We first notice that $h(\bx,\by)\in H^{1,1}=\{0\}\subset H=\{-1,0,1\}$, and thus, this equality constraint 
is one-level. According to Cor.~\ref{Cor: 3}, a corresponding VIP can be derived using \eqref{eq: 1-leve vIP} with $i=1$ 

\begin{align*}
P(1,\bx,\by)=(h(\bx,\:\by)-h_{1})^2= h^2(\bx,\:\by)=&\:(x_ix_j- y_{ij})^2 \\
\underset{TR3(1)}\Longrightarrow =&\: x_ix_j-2x_ix_jy_{ij} +y_{ij}. \label{eq: 1-leve vIP1}
\end{align*}

\noindent This VIP $P(1,\bx,\by)$ is cubic and requires a single ancillary binary variable for quadratization (see TR4 in Section \ref{ssec:TR0-TR4}). \\

\noindent \textit{Approach 2.}  The constraint \eqref{eq: yij=xixj} allows quadratization
in its original decision space through a Rosenberg VIP, which can also be derived using MLCTS as follows. We convert the BOP to an equivalent BLP and then apply MLCTS to derive the corresponding (augmented) QVIP. 
Given a polynomial nonlinear constraint(s) involving variables $\{x_i\}_{i\in I}$, where $I\subset \overline{1,n}$ and $|I|$ does not depend on $n$, the general procedure is as follows:
\begin{enumerate}
\item From $\bbB^{|I|}$, write out the feasible domain $E$ of the  constraint(s);
\item Derive linear constraints describing the polytope $P= conv\: E$;
\item Replace the nonlinear constraint(s) by the analytic description of $P$;
\item Using MLCTS, derive an (augmented) QVIP for each of the linear constraints different from the box constraints $0\le x_i\le 1, i\in I$;
\item Combine all the penalties with positive parameters and form the penalty function for a QUBO model.
\end{enumerate}

        
\noindent We illustrate this process on \eqref{eq: yij=xixj}. 
This constraint involves three binary variables, hence $|I|=3$. Additionally,          
$E=\{(x_i,x_j,y_{ij})\in \bbB^3: (x_i,x_j,y_{ij})\text{ satisfies }\eqref{eq: yij=xixj}\}=\{(0,0,0),(0,1,0),(1,0,0),(1,1,1)\}$, it is a vertex set of a simplex.
\begin{equation*}\label{eq: P simplex 0}
P:\:  y_{ij}\ge 0,\;
 y_{ij}-x_i\le 0,\:
y_{ij}-x_j\le 0,\:
x_i+x_j-y_{ij}\le 1.
\end{equation*} 


\noindent The last inequality can be written as $-1 \le x_i+x_j-y_{ij}\le 1$ and further refined to $0\le x_i+x_j-y_{ij}\le 1$ as the only binary vector $(x_i,x_j,y_{ij})=(0,0,1)$ with $x_i+x_j-y_{ij}=-1$ does not belong to $P$. 
Thus, $P$ can be effectively defined by the following four 2-level regular constraints:
\begin{equation*}\label{eq: P simplex}
0\le y_{ij}\le 1,\;
-1\le y_{ij}-x_i\le 0,\;
-1\le y_{ij}-x_j\le 0,\;
0\le x_i+x_j-y_{ij}\le 1.
\end{equation*} 


\noindent Applying \eqref{eq: QVIP reg} to the last three constraints with $(j=-1,\: I^+=\{(i,j)\},\:I^-=\{i\})$,  $(j=-1,\: I^+=\{(i,j)\},\:I^-=\{j\})$ and  $(j=0,\: I^+=\{(i,j)\},\:I^-=\{i,j\})$, respectively
we come to a family of QVIPs for violating the constraint $(x_i, x_j , y_{ij})\in E $. Combining them with associated penalty parameters $\lambda_1,\lambda_2,\lambda_3$, we obtain the following penalty function.
\begin{eqnarray*}
P_{\lambda}(x_i,x_j,y_{ij})=\lambda_1(-x_iy_{ij}+y_{ij})+\lambda_2(-x_jy_{ij}+y_{ij})+\lambda_3(x_ix_j-x_iy_{ij}-x_jy_{ij}+y_{ij}),
\end{eqnarray*}
where $\mathbf{\lambda}=(\lambda_1,\lambda_2,\lambda_3)\in \bbR^3_{>0}$.
In particular, $P_{(1,1,1)}(x_i,x_j,y_{ij})=x_ix_j-2x_iy_{ij}-2x_jy_{ij}+3y_{ij}$
is exactly the Rosenberg penalty (see \eqref{eq: Rosen}) for the constraint 
\eqref{eq: yij=xixj}.

Note that this BOP-to-BLP transformation scheme can be applied to those BOPs where $f_i(\bx),\: i\in \overline{1,m'}$ in \eqref{eq: BOP eq constr} and \eqref{eq: BOP ineq constr} are arbitrary functions, not necessarily polynomial. This extends the area of application of MLCTS significantly.

Several questions remain open, such as identifying tight bounds  $\underline{b},\overline{b}$ of a function $h(\bx)$ on a feasible domain of a BOP and searching for all values of the function on the domain. Another critical issue is that, if the number of $h(\bx)$-values is exponential in the BOP dimension, instead of brute-force enumeration, a more effective approach (\eg a polynomial algorithm) for the quadratization of the constraint is needed. Conclusive answers to these questions are beyond the scope of this paper, but they merit further exploration in future work.


%% file: sec7-conclusion.tex
\section{Conclusion} \label{sec: conclusion}

In connection with recent successes in quantum computing, especially the development of quantum hardware, the issue of expanding the applications of Quadratic Unconstrained Optimization Problems (QUBO) is becoming increasingly relevant. 
%
Conventionally, such research focuses on building or improving QUBO models for a specific class of combinatorial optimization problems. In contrast, we study the ways of reformulating the general binary optimization problem (BOP) as a QUBO and propose a novel QUBO modelling scheme, the Multilevel Constraint Transformation Scheme (MLCTS). 

Utilizing the concept of constraint levelness, MLCTS provides a general-purpose BLP-to-QUBO reformulation and reduces the number of ancillary binary variables by at least the number of inequality constraints. We formally specify the use of MLCTS on BLP models and derive sufficient conditions for the existence of a compact QUBO formulation (Section \ref{sec: MLCTS}). The generalization of MLCTS to BOP models are also outlined along with an illustration (Section \ref{sec: discussion}).
%


MLCTS is successfully applied on synthetic numerical examples (Section \ref{sec: MLCTS developing}) and well-known combinatorial optimization problems, Maximum 2-Satisfibility Problem (Max2Sat), Linear Ordering Problem (LOP), Community Detection Problem (CDP) and the Maximum Independence Set Problem (MIS) (Section \ref{sec: applied problems}). For proof-of-concept, we conduct computational experiments to evaluate the performance of two QUBO models of MIS problems using a general-purpose solver, Gurobi and a QUBO-specialized solver, the Fujitsu Digital Annealer. As shown in the computational results, the compact QUBO model derived by MLCTS outperforms the augmented QUBO model derived from the conventional transformation scheme on both solvers.

%% file: sec8-appendix.tex
\section{Appendix - Illustration of Transportation 6}\label{sec:appendix-higher-level}
\paragraph{Transportation 6 Variant 1}
Here, we first illustrate the Transportation 6 variant 1 through an exmaple with $k=3$.

Given \eqref{eq: phi1,2}, \eqref{eq: P1,2}, we have
\begin{align*}
    &\left\{\begin{aligned}
        \varphi_1\left( \bx,\bs \right) =& h_0(\bx)-h^0_3s_3, \\
     \varphi_2\left( \bx,\bs \right) =& h_0(\bx)+(1-h^0_3)s_3\\
    \end{aligned}\right. \\
    \Longrightarrow&\;\left\{\begin{aligned}
    P^1\left( \bx,\bs \right) =& (h_0(\bx)-h^0_3s_3)(h_0(\bx)-h^0_3s_3-1)\\
    P^2\left( \bx,\bs \right) =& (h_0(\bx)+(1-h^0_3)s_3)(h_0(\bx)+(1-h^0_3)s_3-1)
    \end{aligned}\right.
\end{align*}

We then apply \eqref{eq: (1.1)}.
\begin{align}
P(\bx,\bs) =&\; (h_0(\bx)-h^0_3s_3)(h_0(\bx)-h^0_3s_3-1) \nonumber\\
            &\; +(h_0(\bx)+(1-h^0_3)s_3)(h_0(\bx)+(1-h^0_3)s_3-1)\nonumber\\
\underset{TR3(2)}\Longrightarrow\; P(\bx,\bs) =&\; (h^0_3)^2s_3-2h_0(\bx)h^0_3s_3+h_0(\bx)^2+h_0(\bx)s_3-h_0(\bx)\nonumber
\\
P(\bx,\bs) =&\;h_0(\bx)^2+h_0(\bx)(s_3-2h^0_3s_3-1)+(h^0_3)^2s_3.\label{eq: Tr5.1n=3}
\end{align}

\paragraph{Transportation 6 Variant 2}

Here, we first illustrate the Transportation 6 variant 2 through an example with $k=3$.
%
With $k=3$, we have $n'=\left\lceil ln_23 \right\rceil=2 $, $a=(3-1)-2^0=1$ and  ${{\varphi }_{3}}\left( \bx,\bs \right)={{h}_{0}}\left( \bx \right)+s_1$. Therefore, 
\begin{align}
{{P}_{{3}}}\left( \bx,\bs \right)=&\;\left( {{\varphi }_{3}}\left( \bx,\bs \right)-3+1 \right)\left( {{\varphi }_{3}}\left( \bx,\bs \right)-3+2 \right)\nonumber\\
=&\;({{h}_{0}}\left( \bx \right)+s_1-2)({{h}_{0}}\left( \bx \right)+s_1-1)\nonumber\\
\underset{TR3(1)}\Longrightarrow {{P}_{{3}}}=&\;(h_0(\bx))^2+h_0(\bx)(2s_1-3)-3s_1+2\label{eq: P TR5.2}
\end{align}

\section{Appendix - MLCTS Illustration} \label{sec:MLCTS-illustration}
\newcommand{\midsepremove}{\aboverulesep = 0mm \belowrulesep = 0mm}  \newcommand{\midsepdefault}{\aboverulesep = 0.605mm \belowrulesep = 0.984mm}

\midsepremove

Consider the following constraint, we demonstrate different ways of deriving VIPs with MLCTS.
\begin{align}
-1\le x_1+x_2-x_3-2x_4\le 1 \label{eq: ex2 constr}
\end{align}

\noindent Consequently, we have $h(\bx)=x_1+x_2-x_3-2x_4$, $H=\{-2,-1,0,1,2\}$, $K=3$, $h(\bx)\in H^{2,3}=\{-1,0,1\}$, i.e. $i=2$, $k=3$. Lastly, $[\underline{b},\overline{b}]\subset (h_1, h_K)$, that is, \eqref{eq: 2-sided} is fulfilled.  Hence, we need a quadratization of the two-sided constraint with odd $k$.

Following MLCTS, a polynomial VIP is obtained by \eqref{eq: 2-side,ikj}. With $j=2$, we have
\begin{align}
P(2,\bx) &=\;(h(\bx)-h_2)(h(\bx)-h_3)^2(h(\bx)-h_4) \nonumber\\
=&\;(x_1+x_2-x_3-2  x_4+1)(x_1+x_2-x_3-2  x_4)^2  (x_1+x_2-x_3-2  x_4-1)\\
\underset{TR3(1)}\Longrightarrow P(\bx)&=\;4  x_1  x_2  x_3  x_4-x_1  x_2  x_3-4  x_1  x_3  x_4-4  x_2  x_3  x_4+x_1  x_2-x_1  x_4-x_2  x_4\nonumber \\ &\;+5  x_3  x_4+x_4 \label{eq: cubic P(x)}
\end{align}

One may easily verify that $P(\bx)$ is indeed a VIP for this constraint as it is strictly positive at unfeasible points and zero at all feasible. However, it is a bi-quadratic VIP,  

In the rest of this section, we show four different ways of quadratization of $P(\bx)$  by using MLCTS. We refer to them as Approach 1-Approach 4 ($A_1-A_4$). Also, we use notation $P^{A_i'}(\bx,\bs)$ for the AQVIP built from $P(\bx)$ by approach $A_i$, $i\in \overline{1,4}$. Respectively, $P^{A_i}(\bx)=\min_{\bs\in \bbB^{n'}}P^{A_i'}(\bx,\bs)$ will be a VIP for the constraint \eqref{eq: ex2 constr}. For each approach, we also show the validation in Table~\ref{tab:Approach 1}-\ref{tab:Approach 34} where it is tabulated on the whole search domain $\bbB^4$). In a column $+/-$, $+$ corresponds to feasible binary points of \eqref{eq: ex2 constr} and $-$ corresponds to unfeasible ones. 

\subsection{Approach 1: MLCTS(TR4.1) }

For quadratization of $P(\bx)$, the following two substitution 
\begin{eqnarray}\label{eq: subst}
y_{1,2}=x_1x_2,\; y_{3,4}=x_3x_4
\end{eqnarray} 
 are sufficient. We apply TR4.1 to $ P(\bx)$ and obtain the following.

\begin{eqnarray*}\label{eq: phi(x,y)}
\begin{aligned}
P(\bx) \underset{TR4.1}\Longrightarrow \varphi(\bx,\by)=&\;x_1x_2-x_1x_4-4x_1y_{3,4} \\
&\;-x_2x_4-4x_2y_{3,4}+5x_3x_4-x_3y_{1,2}+4y_{1,2}y_{3,4}+x_4
\end{aligned}
\end{eqnarray*}

$\varphi(\bx,\by)$ is in a lifted space $\bx,\by\in \bbB^6$, where $\by=(y_{1,2},y_{3,4})$. It is easy to check that 
on a search domain $\bbB^6$, $\varphi(\bx,\by)$ takes values of different signs (see Table~\ref{tab:Approach 1} in Appendix).
Along with substitution \eqref{eq: subst}, the following Rosenberg penalties must be incorporated into the objective.
\begin{eqnarray*}
R_{1,2} = x_1 x_2 - 2 x_1 y_{1,2} - 2 x_2 y_{1,2} + 3 y_{1,2}\\
R_{3,4} = x_3x_4 - 2x_3y_{3,4} - 2 x_4y_{3,4} + 3 y_{3,4}
\end{eqnarray*}

There exists a penalty parameter $\lambda^*> 0$ such that, $\forall\lambda\ge \lambda^*$, $P'_\lambda(\bx,\by)=\varphi(\bx,\by)+\lambda(R_{1,2}+R_{3,4})$ is an QAVIP for violating constraint \eqref{eq: ex2 constr}. 
We use $\lambda = 4$, and then obtain a QAVIP whose validity can be easily checked.

\begin{align*}
P^{A_1'}(\bx,\by)=&\; P'_4(\bx,\by) \\
=&\; 5x_1x_2-x_1x_4-8x_1y_{1,2}-4x_1y_{3,4}-x_2x_4-8x_2y_{1,2}-4x_2y_{3,4}+9x_3x_4\\
&\;-x_3y_{1,2}-8x_3y_{3,4}-8x_4y_{3,4}+4y_{1,2}y_{3,4}+x_4+12y_{1,2}+12y_{3,4}
\end{align*}
Given a QAVIP $P^{A_1'}$, we take the minimum over ancillary variables $\by$ and get a corresponding VIP $P^{A_1}(\bx)=\min_{\by\in \bbB^2}P^{A_1'}(\bx,\by)$ (see Table~\ref{tab:Approach 1}). 
 

\begin{table}[htbp]
\centering
\begin{tabular}{
|c|
>{\centering\arraybackslash}p{.13in}
>{\centering\arraybackslash}p{.13in}
>{\centering\arraybackslash}p{.13in}
>{\centering\arraybackslash}p{.13in}
|c|c|
>{\centering\arraybackslash}p{.13in}
>{\centering\arraybackslash}p{.13in}
>{\centering\arraybackslash}p{.13in}
>{\centering\arraybackslash}p{.13in}|
>{\centering\arraybackslash}p{.13in}
>{\centering\arraybackslash}p{.13in}
>{\centering\arraybackslash}p{.13in}
>{\centering\arraybackslash}p{.13in}|
>{\centering\arraybackslash}p{.13in}|}\hline
& \multicolumn{4}{c|}{$\bx$}& &  & & & & & & & & &\\
   \rotatebox[origin=c]{90}{cases}  &$x_1$&$x_2$&$x_3$&$x_4$
    &\rotatebox[origin=c]{90}{(+/-)} 
        &\rotatebox[origin=c]{90}{$P(\bx)$}
    &\rotatebox[origin=c]{90}{$\varphi(\bx,0,0)$}
        &\rotatebox[origin=c]{90}{$\varphi(\bx,1,0)$}
        &\rotatebox[origin=c]{90}{$\varphi(\bx,0,1)$}
        &\rotatebox[origin=c]{90}{$\varphi(\bx,1,1)$}
    &\rotatebox[origin=c]{90}{$P^{A_1'}(\bx,0,0)$}
        &\rotatebox[origin=c]{90}{$P^{A_1'}(\bx,1,0)$}
            &\rotatebox[origin=c]{90}{$P^{A_1'}(\bx,0,1)$}
                &\rotatebox[origin=c]{90}{$P^{A_1'}(\bx,1,1)$}
    &\rotatebox[origin=c]{90}{$P^{A_1}(\bx)$}\\ \hline
1&0&0&0&0&+&0&0&0&0&4&0&12&12&28&0\\
2&1&0&0&0&+&0&0&0&-4&0&0&4&8&16&0\\
3&0&1&0&0&+&0&0&0&-4&0&0&4&8&16&0\\
4&0&0&1&0&+&0&0&-1&0&3&0&11&4&19&0\\
5&1&1&0&0&-&1&1&1&-7&-3&5&1&9&9&1\\
6&1&0&1&0&+&0&0&-1&-4&-1&0&3&0&7&0\\
7&0&1&1&0&+&0&0&-1&-4&-1&0&3&0&7&0\\
8&1&1&1&0&+&0&1&0&-7&-4&5&0&1&0&0\\
9&0&0&0&1&-&1&1&1&1&5&1&13&5&21&1\\
10&1&0&0&1&+&0&0&0&-4&0&0&4&0&8&0\\
11&0&1&0&1&+&0&0&0&-4&0&0&4&0&8&0\\
12&0&0&1&1&-&6&6&5&6&9&10&21&6&21&6\\
13&1&1&0&1&+&0&0&0&-8&-4&4&0&0&0&0\\
14&1&0&1&1&-&1&5&4&1&4&9&12&1&8&1\\
15&0&1&1&1&-&1&5&4&1&4&9&12&1&8&1\\
16&1&1&1&1&+&0&5&4&-3&0&13&8&1&0&0\\
\hline
\end{tabular}
        \caption{Validation of VIPs derived by Approach 1.}
    \label{tab:Approach 1}
\end{table}

\subsection{Approach 2: MLCTS(TR4.2)}

One positive and three negative monomials are present in the bi-quadratic VIP $P(\bx)$. 
%
  For the positive one (\ie $x_1x_2x_3x_4$), $I=\overline{1,4}$ and formulas \eqref{eq: S1,S2} becomes $S_1(I) = \sum_{i=1}^4 x_i$, $S_2(I) = \sum_{i,j=1; i<j}^4 x_ix_j$. Applying formula \eqref{eq: min selection formula2} to this term we have
\begin{eqnarray}\label{eq: TR4.21}
\begin{aligned}
x_1x_2x_3x_4 \;\Longrightarrow\; \min_{s_0}\big\{\varphi_0(\bx,\bs)= 
 & s_0(-2x_1-2x_2-2x_3-2x_4+3)+x_1x_2 \\ &+x_1x_3+x_1x_4+x_2x_3+x_2x_4+x_3x_4\big\}
\end{aligned}
\end{eqnarray}

For the three negative monomials (\ie $-x_1  x_2  x_3,\: -4  x_1  x_3  x_4, \: -4  x_2  x_3  x_4$),  formula \eqref{eq: min selection formulanegative} is used. Let $I=\{i,j,k\}$, then it becomes  $-x_ix_jx_k= -\min_{s \in \bbB} s (x_i+x_j+x_k-2)$. Applying this to the triple monomials in $P(\bx)$, we have
\begin{align}
&-x_1x_2x_3 \Longrightarrow  -\min_{s_1}\{\varphi_1(\bx,\bs)=s_1(x_1+x_2+x_3-2),\label{eq: TR4.22}\\
&-x_1x_3x_4\Longrightarrow -\min_{s_2}\{\varphi_2(\bx,\bs)=s_2(x_1+x_3+x_4-2), \nonumber\\
&-x_2x_3x_4\Longrightarrow -\min_{s_3}\{\varphi_3(\bx,\bs)=s_3(x_2+x_3+x_4-2)\nonumber
\end{align}

Consequently, we have
\begin{align*}
P(\bx)\underset{TR4.2}\Longrightarrow P^{A_2'}(\bx,\bs)=&\;4  \varphi_0(\bx,\bs)+\varphi_1(\bx,\bs)+4\varphi_2(\bx,\bs)+\varphi_3(\bx,\bs)\\
&\;+x_1  x_2-x_1  x_4-x_2  x_4+5  x_3  x_4+x_4\\
=&\;5x_1x_2+3x_1x_4+3x_2x_4+9x_3x_4+x_4+ 4x_1x_3+4x_2x_3\\
&\;+4s_0(-2x_1-2x_2-2x_3-2x_4+3)-s_1(x_1+x_2+x_3-2)\\
&\;-4s_2(x_1+x_3+x_4-2)-4s_3(x_2+x_3+x_4-2)
\end{align*}

Similar to Approach 1, we obtain a VIP $P^{A_2}(\bx)=\min_{\bs\in \bbB^4} P^{A_2'}(\bx,\bs)$ by taking the minimum over ancillary variables $\bs$. Verification of the validity of formulas \eqref{eq: TR4.21}, \eqref{eq: TR4.22} and that $P^{A_2}(\bx)$ is a VIP for our constraint is given in Table~\ref{tab:Approach 2}.

\begin{table}[htbp]
\centering
    \begin{tabular}{
    |c|
    >{\centering\arraybackslash}p{.13in}
    >{\centering\arraybackslash}p{.13in}
    >{\centering\arraybackslash}p{.13in}
    >{\centering\arraybackslash}p{.13in}|c|
    >{\centering\arraybackslash}p{.13in}
    >{\centering\arraybackslash}p{.13in}
    >{\centering\arraybackslash}p{.13in}
    >{\centering\arraybackslash}p{.13in}|
    >{\centering\arraybackslash}p{.13in}
    >{\centering\arraybackslash}p{.13in}
    >{\centering\arraybackslash}p{.13in}
    >{\centering\arraybackslash}p{.13in}|c|}\hline
     & \multicolumn{4}{c|}{$\bx$}& & & & & & & & & &\\
    i &$x_1$&$x_2$&$x_3$&$x_4$& \rotatebox[origin=c]{90}{(+/-)} &\rotatebox[origin=c]{90}{$x_1x_2x_3x_4$}&\rotatebox[origin=c]{90}{$\varphi_0(\bx,0)$}&\rotatebox[origin=c]{90}{$\varphi_0(\bx,1)$}&\rotatebox[origin=c]{90}{$\min_{s_0}\varphi_0(\bx,s_0)$}&\rotatebox[origin=c]{90}{$-x_1x_2x_3$}&\rotatebox[origin=c]{90}{$\varphi_1(\bx,0)$}&\rotatebox[origin=c]{90}{$\varphi_1(\bx,1)$}&\rotatebox[origin=c]{90}{$\min_{s_1}\varphi_1(x,s_1)$}&\rotatebox[origin=c]{90}{$P^{A_2}(\bx)$}\\
     \hline
1&0&0&0&0&+&0&0&3&0&0&0&2&0&0\\
2&1&0&0&0&+&0&0&1&0&0&0&1&0&0\\
3&0&1&0&0&+&0&0&1&0&0&0&1&0&0\\
4&0&0&1&0&+&0&0&1&0&0&0&1&0&0\\
5&1&1&0&0&-&0&1&0&0&0&0&0&0&1\\
6&1&0&1&0&+&0&1&0&0&0&0&0&0&0\\
7&0&1&1&0&+&0&1&0&0&0&0&0&0&0\\
8&1&1&1&0&+&0&3&0&0&-1&0&-1&-1&0\\
9&0&0&0&1&-&0&0&1&0&0&0&2&0&1\\
10&1&0&0&1&+&0&1&0&0&0&0&1&0&0\\
11&0&1&0&1&+&0&1&0&0&0&0&1&0&0\\
12&0&0&1&1&-&0&1&0&0&0&0&1&0&6\\
13&1&1&0&1&+&0&3&0&0&0&0&0&0&0\\
14&1&0&1&1&-&0&3&0&0&0&0&0&0&1\\
15&0&1&1&1&-&0&3&0&0&0&0&0&0&1\\
16&1&1&1&1&+&1&6&1&1&-1&0&-1&-1&0\\
\hline
    \end{tabular}
        \caption{Validation of VIPs derived by Approach 2.}
    \label{tab:Approach 2}
\end{table}

\subsection{Approach 3: MLCTS(TR0.1+TR6.1)}

We first perform normalization as below (see \eqref{H0ik} and \eqref{eq: h0(x)}), and know that $h^0_3=2$

\begin{align*}
&H^{2,3}=\{-1,0,1\} \underset{\eqref{H0ik}}\Longrightarrow H^{0,2,3}=\{0,1,2\}\\
&h(\bx)=x_1+x_2-x_3-2x_4 \underset{\eqref{eq: h0(x)}}\Longrightarrow h_0(\bx)=x_1+x_2-x_3-2x_4+1
\end{align*}

We substitute $h_0(\bx), h^0_3$ into \eqref{eq: Tr5.1n=3}, and obtain the follows.
\begin{eqnarray*}
\begin{aligned}
P(\bx)\underset{TR0.1+TR6.1}\Longrightarrow  P^{A_3'}(\bx,\bs)=&\; 2x_1x_2-2x_1x_3-4x_1x_4-2x_2x_3\\
&\;-4x_2x_++x_3x_4-3s_3x_1-3s_3x_2+3s_3x_3\\
&\;+6s_3x_4+s_3+2x_1+2x_2+2x_4
\end{aligned}
\end{eqnarray*}

$P^{A_3'}(\bx,\bs)$ is a QVIP for \eqref{eq: ex2 constr}. By taking the minimum over ancillary variables $\by$, we obtain a VIP $P^{A_3}(\bx)=\min_{\bs\in \bbB} P^{A_3}(\bx,\bs)$ for this constraint (see Table~\ref{tab:Approach 34}).

\subsection{Approach 4: MLCTS(TR0.2+TR6.2)}
After the same normalization process as the previous approach, we substitute $h_0(\bx), h^0_3$ into \eqref{eq: P TR5.2}, and obtain the following.
\begin{eqnarray*}
\begin{aligned}
P(\bx)\underset{TR0.2n+TR6.2}\Longrightarrow  P^{A_4'}(\bx,\bs)=&\;x_1x_2-x_1x_3-2x_1x_4-x_2x_3-2x_2x_4+2x_3x_4\\
&\;+s_1x_1+s_1x_2-s_1x_3-2s_1x_4+x_3+3x_4
\end{aligned}
\end{eqnarray*}

$P^{A_4'}(\bx,\bs)$ is a QVIP for \eqref{eq: ex2 constr}. Similar to the previous approaches, $P^{A_4}(\bx)=\min_{\bs\in \bbB} P^{A_4'}(\bx,\bs)$ is a VIP for this constraint 
(see Table~\ref{tab:Approach 34}).

In summary, we present four ways to derive QVIPs with our MLCTS. The number $n'$ of ancillary binary variables in the QVIPs $P^{A_1'}(\bx,\bs)-P^{A_4'}(\bx,\bs)$ are $2,4,1,1$, respectively. The approaches $A_3,\: A_4$ associated with transformation TR6 are the best in this regard and utilize a single ancillary variable. Furthermore, we analyze the range of the values  of $P^{A_3}(\bx)$ and $P^{A_4}(\bx)$, which are $\{0,1,2,6\}$ and $\{0,1,3\}$, respectively. The narrowest range corresponds to $A_4$, implying that QVIP $P^{A_4'}(\bx,\bs)$ is better in this regard.

\begin{table}[htbp]
\centering
\begin{tabular}{
|c|
>{\centering\arraybackslash}p{.14in}
>{\centering\arraybackslash}p{.14in}
>{\centering\arraybackslash}p{.14in}
>{\centering\arraybackslash}p{.14in}
|>{\centering\arraybackslash}p{.14in}|
>{\centering\arraybackslash}p{.14in}
>{\centering\arraybackslash}p{.14in}
>{\centering\arraybackslash}p{.14in}|
>{\centering\arraybackslash}p{.14in}
>{\centering\arraybackslash}p{.14in}
>{\centering\arraybackslash}p{.14in}|}\hline
& \multicolumn{4}{c|}{$\bx$}
& 
& \multicolumn{3}{c|}{Approach 3} 
& \multicolumn{3}{c|}{Approach 4}\\[5pt] 
  i  & $x_1$ & $x_2$ & $x_3$ & $x_4$ & \rotatebox[origin=c]{90}{(+/-)}
    &\rotatebox[origin=c]{90}{$P^{A_3'}(\bx,0)$}
        &\rotatebox[origin=c]{90}{$P^{A_3'}(\bx,1)$}
        &\rotatebox[origin=c]{90}{$P^{A_3}(\bx)$}
    &\rotatebox[origin=c]{90}{$P^{A_4'}(\bx,0)$}
        &\rotatebox[origin=c]{90}{$P^{A_4'}(\bx,1)$}
        &\rotatebox[origin=c]{90}{$P^{A_4}(\bx)$}\\
      \hline
1&0&0&0&0&+&0&1&0&0&0&0\\
2&1&0&0&0&+&2&0&0&0&1&0\\
3&0&1&0&0&+&2&0&0&0&1&0\\
4&0&0&1&0&+&0&4&0&1&0&0\\
5&1&1&0&0&-&6&1&1&1&3&1\\
6&1&0&1&0&+&0&1&0&0&0&0\\
7&0&1&1&0&+&0&1&0&0&0&0\\
8&1&1&1&0&+&2&0&0&0&1&0\\
9&0&0&0&1&-&2&9&2&3&1&1\\
10&1&0&0&1&+&0&4&0&1&0&0\\
11&0&1&0&1&+&0&4&0&1&0&0\\
12&0&0&1&1&-&6&16&6&6&3&3\\
13&1&1&0&1&+&0&1&0&0&0&0\\
14&1&0&1&1&-&2&9&2&3&1&1\\
15&0&1&1&1&-&2&9&2&3&1&1\\
16&1&1&1&1&+&0&4&0&1&0&0\\
\hline
    \end{tabular}
        \caption{Validation of VIPs derived by Approach 3 and 4.}
    \label{tab:Approach 34}
\end{table}

\vspace{1em}

\section{Appendix - Validation of CDP Quadratization}

\begin{table}[htbp]
\centering
\begin{subtable}{\textwidth}
\centering
\begin{tabular}{
|c|
>{\centering\arraybackslash}p{.3in}
>{\centering\arraybackslash}p{.3in}
>{\centering\arraybackslash}p{.3in}|
c|
}\hline
cases&$x_{uv}$&$x_{vw}$&$x_{uw}$&$P'_{u,v,w}(X)$\\
  \hline
1&0&0&0&0\\
2&1&0&0&1\\
3&0&1&0&1\\
4&0&0&1&1\\
5&1&1&0&0\\
6&1&0&1&0\\
7&0&1&1&0\\
8&1&1&1&0\\
\hline
\end{tabular}
    \caption{Validation of the cubic VIP $P'_{u,v,w}(X)$.}
    \label{tab: 0.2}
\end{subtable}

\begin{subtable}{\textwidth}
\centering
\begin{tabular}{
|c|
>{\centering\arraybackslash}p{.3in}
>{\centering\arraybackslash}p{.3in}
>{\centering\arraybackslash}p{.3in}
>{\centering\arraybackslash}p{.3in}|
c|
}\hline
cases&$x_{uv}$&$x_{vw}$&$x_{uw}$&$y_{uvw}$&$P'''_{u,v,w}(X,Y)$\\
  \hline
1&0&0&0&0&0\\
2&1&0&0&0&1\\
3&0&1&0&0&1\\
4&0&0&1&0&1\\
5&1&1&0&0&4\\
6&1&0&1&0&0\\
7&0&1&1&0&0\\
8&1&1&1&0&1\\
9&0&0&0&1&12\\
10&1&0&0&1&5\\
11&0&1&0&1&5\\
12&0&0&1&1&16\\
13&1&1&0&1&0\\
14&1&0&1&1&7\\
15&0&1&1&1&7\\
16&1&1&1&1&0\\
\hline
\end{tabular}
    \caption{Validation of the QAVIP $P'''_{u,v,w}(X,Y)$.}
    \label{tab: 0.3}
\end{subtable}
 \caption{Verification of the validity of CDP quadratization (See Section \ref{ssec:CDP-QUBO}).}
\end{table}